\documentclass[12pt,reqno]{amsart}
\usepackage{amssymb,amsthm,amsmath,amsxtra,dsfont,appendix,bbm,stix}
\usepackage{hyperref} 
\usepackage{mathrsfs}

\numberwithin{equation}{section}
\makeatletter

\makeatletter
\def\@tocline#1#2#3#4#5#6#7{\relax
  \ifnum #1>\c@tocdepth 
  \else
    \par \addpenalty\@secpenalty\addvspace{#2}%
    \begingroup \hyphenpenalty\@M
    \@ifempty{#4}{%
      \@tempdima\csname r@tocindent\number#1\endcsname\relax
    }{%
      \@tempdima#4\relax
    }%
    \parindent\z@ \leftskip#3\relax \advance\leftskip\@tempdima\relax
    \rightskip\@pnumwidth plus4em \parfillskip-\@pnumwidth
    #5\leavevmode\hskip-\@tempdima
      \ifcase #1
       \or\or \hskip 1em \or \hskip 2em \else \hskip 3em \fi%
      #6\nobreak\relax
      \dotfill
      \hbox to\@pnumwidth{\@tocpagenum{#7}}
    \par
    \nobreak
    \endgroup
  \fi}
\makeatother

\newcommand{\cO}{\mathcal{O}}

\renewcommand{\i}{\mathrm{i}}

\newcommand{\B}{\mathfrak{b}}
\newcommand{\C}{\mathds{C}}
\newcommand{\E}{\mathds{E}}
\renewcommand{\P}{\mathds{P}}
\renewcommand{\d}{\mathrm{d}}
\newcommand{\1}{\mathds{1}}
\newcommand{\cqh}{\mathrm{c}_{\rm qh}}
\newcommand{\Psiqh}{\Psi_{\rm qh}}
\newcommand{\K}{\mathrm{K}}
\renewcommand{\O}{\mathcal{O}}
\newcommand{\M}{\mathbf{M}}
\newcommand{\N}{\mathbb{N}}
\newcommand{\R}{\mathds{R}}
\newcommand{\bR}{\mathbf{R}}

\renewcommand{\S}{\mathfrak{S}}

\newcommand{\w}{\mathbf{w}}

\newcommand{\z}{\mathbf{z}}

\newcommand{\D}{\mathds{D}}

\newtheorem{theorem}{Theorem}[section]
\newtheorem{conjecture}[theorem]{Conjecture}
\newtheorem{heuristic}[theorem]{Heuristic}
\newtheorem{notation}[theorem]{Notation}

\newtheorem{proposition}[theorem]{Proposition}
\newtheorem{corollary}[theorem]{Corollary}
\newtheorem{lemma}[theorem]{Lemma} 
\newtheorem{remark}[theorem]{Remark} 
\newtheorem{properties}[theorem]{Properties}

\newcommand{\tr}{\mathrm{Tr}}

\newcommand{\LLL}{\mathrm{LLL}}

\newcommand{\bx}{\mathbf{x}}
\newcommand{\by}{\mathbf{y}}
\newcommand{\bz}{\mathbf{z}}
\newcommand{\bw}{\mathbf{w}}
\newcommand{\gH}{\mathfrak{H}}
\newcommand{\norm}[1]{\left\lVert #1 \right\rVert}
\newcommand{\ang}{\mathrm{arg}}
\newcommand{\bA}{\mathbf{A}}
\newcommand{\bAt}{\mathbf{A}^{\mathrm{tot}}}
\newcommand{\cA}{\mathcal{A}}
\newcommand{\cL}{\mathcal{L}}
\newcommand{\cV}{\mathcal{V}}
\newcommand{\cZ}{\mathcal{Z}}

\newcommand{\nubf}{\boldsymbol{\nu}}
\newcommand{\cE}{\mathcal{E}}

\newcommand{\cD}{\mathscr{D}}

\newcommand{\rhot}{\varrho^{\mathrm{tot}}}

\newcommand{\curl}{\mathrm{curl}}

\newcommand{\DROPn}{\mathscr{D}_n}
\newcommand{\Phit}{\widetilde{\Phi}}
\newcommand{\Heff}{H^{\rm eff}}
\newcommand{\Eeff}{E^{\rm eff}}

\numberwithin{equation}{section}

\begin{document}

\title{Quantum statistics transmutation via magnetic flux attachment}

\author[G. Lambert]{Gaultier Lambert}
\address{University of Zurich}
\email{gaultier.lambert@math.uzh.ch}

\author[D. Lundholm]{Douglas Lundholm}
\address{Department of Mathematics, Uppsala University, Box 480, SE-751 06, Uppsala, Sweden}
\email{douglas.lundholm@math.uu.se}

\author[N. Rougerie]{Nicolas Rougerie}
\address{Ecole Normale Sup\'erieure de Lyon \& CNRS,  UMPA (UMR 5669)}
\email{nicolas.rougerie@ens-lyon.fr}

\date{January, 2023}

\begin{abstract}
We consider a model for two types (bath and tracers) of 2D quantum particles in a perpendicular magnetic field. Interactions are short range and inter-species, and we assume that the bath particles are fermions, all lying in the lowest Landau level of the magnetic field. Heuristic arguments then indicate that, if the tracers are strongly coupled to the bath, they effectively change their quantum statistics, from bosonic to fermionic or vice-versa. We rigorously compute the energy of a natural trial state, indeed exhibiting this phenomenon of \emph{statistics transmutation}. The proof involves estimates for the characteristic polynomial of the Ginibre ensemble of random matrices.
\end{abstract}

\maketitle

\tableofcontents

\section{Introduction}

Quantum statistics refers to the ways that identical quantum particles may distribute themselves into different energy configurations or probability distributions 
--- Bose-Einstein statistics obeyed by bosons, respectively Fermi-Dirac statistics obeyed by fermions --- 
and it is rooted in the exchange and correlation properties of the many-particle wave function.
Pauli's exclusion principle is the concrete physical manifestation of the mathematical fact that identical fermions, such as electrons, necessarily occupy individual states as a consequence of the permutation-antisymmetric, or determinantal, structure of their wave function.
The macroscopic consequences of particles being subject to, respectively being exempt of, such exclusion properties include effects such as 
the stability of fermionic matter, chemistry and conduction, respectively Bose-Einstein condensation and optical coherence, and can thus be said to
dominate our experience of the three-dimensional world. 

On the other hand, quasi-particles (or quasi-holes) excitations of 2D quantum systems in high magnetic fields are expected to exhibit unusual physics~\cite{Jain-07,Laughlin-99,Myrheim-99,Wilczek-90}. 
In somewhat imaged wording, the mechanism is as follows. The combination of interparticle repulsion (due to the Pauli principle or genuine, e.g. Coulomb, interactions) and strong magnetic fields (forcing quantized cyclotron motion) leads particles to ``bind to magnetic flux quanta''. They thus effectively behave as producing a magnetic field, constant for a homogeneous system. Impurities/excitations in such a system can experience two noticeable effects:

\medskip

\noindent \boxed{\textbf{A}} They can feel the effective magnetic field, resulting in their coupling (i.e. electric charge) to the actual external magnetic field being modified.

\medskip

\noindent \boxed{\textbf{B}} By piercing holes in the underlying system they may deplete the effective field in a neighborhood around their location, which leads them to also ``bind to magnetic flux quanta''. This can modify their interactions, and in effect change their quantum statistics.  

\medskip

A mathematically rigorous understanding of these phenomena seems an important challenge, in particular in cases where the effect of \boxed{\textbf{B}} is most spectacular: it is known~\cite{AroSchWil-84,ZhaSreGemJai-14,CooSim-15,LunRou-16,UmuMacComCar-18,YakEtal-19,Brooks-etal-21a,Brooks-etal-21b} that the modification of the excitations' statistics can lead to the emergence of quasi-particles with \emph{fractional statistics}~\cite{LeiMyr-77,GolMenSha-81,Wilczek-82b}. Recent experiments~\cite{BarEtalFev-20,NakEtalMan-20} have provided concrete evidence that the above mechanisms lead to the emergence of quasi-particles, termed ``anyons'', violating the dichotomy between bosons and fermions which is obeyed by all fundamental particles.

Our goal in this paper is to make progress towards rigorously establishing Effects \boxed{\textbf{A}} and \boxed{\textbf{B}} in a well-defined model. The emergence of fractional statistics is out of our present reach but we obtain a clear signature of \emph{statistics transmutation}: bosonic (respectively fermionic) impurities effectively behave as fermions (bosons), acquiring (loosing) a Pauli exclusion principle via their coupling to a specific quantum bath.  

Our model is made of the following main ingredients:
\begin{itemize}
 \item a large bath of non-interacting 2D fermions,
 \item several tracer particles (impurities) immersed in this bath,
 \item a strong short-range repulsive interaction between bath and tracer particles,
 \item an external magnetic field to which both types of particles may couple. 
\end{itemize}

These ingredients pose severe restrictions to any practical realization, and the most reasonable venue would be in gases of cold atoms, where similar systems have already been proposed~\cite{ZhaSreGemJai-14,ZhaSreJai-15,CooSim-15,LunRou-16,WuJai-15}. The magnetic field would then be an artifical one~\cite{BloDalZwe-08,DalGerJuzOhb-11}, or brought about by rotation~\cite{Cooper-08,Fetter-09,Viefers-08,LMS19,KMS21,SLMS21}. 
From now on we consider the above system as a thought experiment
that can be used to extract the features of interest we may expect of a realistic physical system, 
and proceed to describe it in mathematical terms. 

\subsection{Model}

We assume that the external magnetic field is strong enough to force the bath particles to all live in the ground energy level of their magnetic kinetic energy (i.e. the lowest Landau level, $\LLL$). The one-body Hilbert space for the bath particles is thus~\cite{RouYng-19} 
\begin{equation}\label{eq:intro LLL}
\LLL:= \left\{ \psi \in L^2 (\R^2) : \psi(x) = f(z) e ^{-\frac{\B}{2}|z|^2}, f \mbox{ analytic} \right\},
\end{equation}
where we identify $\R^2 \ni x \leftrightarrow z \in \C$ and $\B>0$ is the strength of the external magnetic field, in units with Planck's constant $\hbar =1$ and the effective charge (coupling to the magnetic field) is $e=-2$. The above space is the ground eigenspace of the magnetic Laplacian 
\begin{equation}\label{eq:mag Lap}
\left( -\i \nabla_x - \B x^{\perp} \right)^2 
\end{equation}
with constant magnetic field (downward directed) of magnitude
$\B = \frac{\B}{2} \curl \, x^{\perp}$, and $x^\perp = (x^1,x^2)^{\perp} := (-x^2,x^1)$. 

For the tracer particles we consider the full $L^2 (\R^2)$ as one-body Hilbert space. Taking $N$ fermionic particles for the bath and $n$ impurity particles for the tracers, our full Hilbert space for the joint system is thus  
\begin{equation}\label{eq:intro full space}
\gH_{\rm sym}^{n \oplus N} :=  (L^2 (\R^2))^{\otimes_{\rm sym} n} \otimes \LLL ^{\otimes_{\mathrm{asym}} N } = L^2_{\rm sym} (\R^{2n})  \otimes \LLL ^{\otimes_{\mathrm{asym}} N } 
\end{equation}
or
\begin{equation}\label{eq:intro full space bis}
\gH^{n \oplus N}_{\rm asym} :=  (L^2 (\R^2))^{\otimes_{\rm asym} n} \otimes \LLL ^{\otimes_{\mathrm{asym}} N } = L^2_{\rm asym} (\R^{2n})  \otimes \LLL ^{\otimes_{\mathrm{asym}} N } 
\end{equation}
depending on whether the tracers are bosons or fermions. We denote 
$$\Psi_{n \oplus N} (y_1,\ldots,y_n;x_1,\ldots,x_N)$$
a generic element of such a space. This is an $L^2$ function symmetric or antisymmetric under the exchange of its first $n$ arguments and antisymmetric under the exchange of its last $N$ arguments, with in addition $e^{\frac{\B}{2}\sum_{j=1}^N |x_j|^2} \Psi_{n \oplus N}$ being analytic as a function of $x_1,\ldots,x_N$,
which we identify with complex numbers $z_1,\ldots,z_N$. 
Henceforth we also identify the coordinates $y_1,\ldots, y_n$ of the tracer particles with complex numbers $w_1,\ldots,w_n$. We also write 
\begin{equation}\label{eq:complex}
\R^{2n} \ni \by \leftrightarrow \w=(w_1, \dots , w_n) \in\C^n \mbox{ and } \R^{2N} \ni \bx \leftrightarrow \z=(z_1,\dots, z_N)\in\C^N
\end{equation}
for short, with the corresponding Lebesgue measures denoted $\d \z,\d \w$.

The joint Hamiltonian of the system we consider is 
\begin{equation}\label{eq:intro hamil}
H_{n \oplus N}^W := g \sum_{k=1} ^N \sum_{j=1} ^n \delta (x_k - y_j) + \sum_{j=1}^n \left\{ \frac{1}{2m}\left( -\i \nabla_{y_j} - q \B y_j ^{\perp} \right)^2 \right\} + W(y_1,\ldots,y_n)
\end{equation}
with $g>0$ a coupling constant and $W$ a collection of potentials (typically one- and two-body) acting on tracers, whose mass we denote $m$. The ``charge'' of the tracer particles is $-2q$ in our convention. The operator $\delta (x_k - y_j)$ is a delta interaction between bath and tracers, whose precise definition is given in Section~\ref{sec:delta int}. As a quadratic form it acts as 
\begin{equation}\label{eq:intro delta}
\left\langle \Psi_{1 \oplus 1} | \delta (x - y) | \Psi_{1 \oplus 1} \right\rangle_{L^2} = \int_{\R^2} |\Psi_{1 \oplus 1} (x;x)| ^2 dx,
\end{equation}
which is well-defined for $\Psi_{1 \oplus 1}(y;x)\in \gH^{1 \oplus 1}$, using the high regularity in the $x$ variable. 
The remaining terms of \eqref{eq:intro hamil} are considered in the sense of forms as well, 
and the total energy of an $L^2$-normalized state $\Psi_{n \oplus N}$ is given by the associated functional
\begin{equation}\label{eq:energy func}
\cE[\Psi_{n \oplus N}] := \langle \Psi_{n \oplus N} | H_{n \oplus N}^W | \Psi_{n \oplus N} \rangle_{\gH^{n \oplus N}}.
\end{equation}
We are in particular interested in the ground state energy
\begin{equation}\label{eq:energy gs}
E(n \oplus N) := \inf \bigl\{ \cE[\Psi_{n \oplus N}] : \Psi_{n \oplus N} \in \gH^{n \oplus N}_{\rm sym/asym}, \ \int_{\R^{2(n+N)}} |\Psi_{n \oplus N}|^2 = 1 \bigr\}.
\end{equation}
We assumed the bath particles' individual energy to be frozen by the projection to the lowest Landau level (whence the absence of a kinetic energy acting on the $\bx$ degrees of freedom in~\eqref{eq:intro hamil}). We expect that the above can be derived in a suitable limit from a more general model with the help of methods from~\cite{LewSei-09,SeiYng-20}, but do not pursue this here. 

It is convenient to choose a particular scaling of $\B$ and $N$. Indeed, the degeneracy per unit area~\cite{ChaFlo-07,RouYng-19} of a Landau level is well-known to be of order $\B$, the magnetic flux per unit area. Hence the $N$ fermionic bath particles, which must occupy orthogonal states in the lowest Landau level, will fill an area of order $N/\B$. For convenience, we fix this area by choosing 
\begin{equation}\label{eq:B is N}
\B = N \qquad \text{(where $N \to \infty$)}
\end{equation}
in the sequel. One can always reduce to this case by scaling lengths, at the price of modifying parameters in~\eqref{eq:intro hamil} appropriately. 

\subsection{Main result}

We are interested in the ground state problem, namely the lowest eigenvalue $E(n \oplus N)$ of $H_{n \oplus N}^W$ acting on $\gH^{n \oplus N}$. We investigate the parameter regime where effects \boxed{\textbf{A}} and \boxed{\textbf{B}} are expected to take place. Namely we assume that the bath/tracer interaction set by $g>0$ gives the main energy scale in~\eqref{eq:intro hamil}. 
When $g \to \infty$, it is a good approximation to restrict available states to those of the kernel of the interaction operator, the most simple~\cite{LunRou-16,ZhaSreGemJai-14,WuJai-15} being of the form (cf.~Section~\ref{sec:delta int})
\begin{equation} \label{def:psi}
\Psi_{\Phi}(\by;\bx)  = \Phi(\by) \cqh(\w) \Psi_{\rm qh}(\w;\z) 
\end{equation}
for a $\Phi\in L^2 (\R^{2n})$ describing the motion of the tracer particles (hence symmetric under particle exchange in case~\eqref{eq:intro full space} and antisymmetric in case~\eqref{eq:intro full space bis}). Recall our convention~\eqref{eq:complex} of identifying real two-vectors and complex numbers. We here set 
\begin{equation}\label{eq:norm Phi}
\int_{\R^{2n}} |\Phi| ^2 = 1 
\end{equation}
and we thus ensure normalization of~\eqref{def:psi} by choosing
$\cqh(\w) > 0$ with
\begin{equation} \label{def:cqh}
 \cqh(\w)^{-2} := \int_{\C^N} \left| \Psi_{\rm qh}(\w ; \z)\right|^2 \d\z. 
\end{equation}
The main ingredient of the construction is the function $\Psi_{\rm qh}$, describing a ``filled Landau level with quasi-holes'':
\begin{equation} \label{def:Psiqh}
\Psi_{\rm qh}(\w;\z) : = \left(\prod_{j=1}^n\prod_{k=1}^N (w_j-z_k)\right) \left(\prod_{1\leq k < \ell \leq N} (z_k-z_\ell) \prod_{k=1}^N e^{-\B |z_k|^2/2}\right).
\end{equation}
Our main goal is to compute an effective energy functional for the remaining degree of freedom, i.e. the choice of the function $\Phi$. Clearly, for any potential $W(y_1,\ldots,y_n)$, 
\begin{equation}\label{eq:pot ener} 
\left\langle\Psi_{\Phi} | W(y_1,\ldots,y_n) |\Psi_{\Phi} \right\rangle_{L^2(\R^{2(N+n)})} = \left\langle \Phi | W(y_1,\ldots,y_n) |\Phi \right\rangle_{L^2(\R^{2n})}.
\end{equation}
The main question thus concerns the magnetic kinetic energy of the tracer particles. We assume that they stay within the extension of the bath\footnote{If they leave the bath for some reason they no longer interact with it, and this whole discussion is irrelevant.}. 
With our choice of units 
(in particular~\eqref{eq:B is N} and $n \ll N$) the latter is essentially the unit disk
\begin{equation}\label{eq:unit disk}
\D:= \left\{ x \in \R^2, |x| \leq 1\right\}. 
\end{equation}
To avoid boundary effects we actually consider the slightly smaller
\begin{equation}\label{eq:config}
\DROPn := \big\{ (y_1,\ldots,y_n) \in\R^{2n} : |y_j| \le 1-\delta_N(\kappa) \mbox{ for all } j= 1 \ldots n\big\} 
\end{equation}
as the configuration space for the tracers, where, for a fixed constant $\kappa >0$ we have set 
\begin{equation}\label{eq:delta}
\delta_N(\kappa):= \kappa \sqrt{\frac{\log N}{N}}. 
\end{equation}
One of the most delicate aspects of the analysis is the behavior of the system when two tracer particles get close, $y_i\sim y_j$ for $i\neq j$. It is hence convenient to further define 
\begin{equation}\label{eq:nomerg}
\DROPn^{\varnothing} := \big\{ (y_1,\ldots,y_n) \in \DROPn : |y_i-y_j| \ge 2 \delta_N(\kappa) \text{ for all }i,j=1,\dots,n \text{ with }i\neq j \big\},
\end{equation}
the set where no such merging occurs. 

The statistics transmutation phenomenon is highlighted by setting
\begin{equation}\label{eq:choice Phi}
\Phi (y_1,\ldots,y_n) = \left(\prod_{1\leq i<j \leq n} e^{\i \,\ang(y_i-y_j)} \right) \Phit (y_1,\ldots,y_n)
\end{equation}
with $\ang(y_i-y_j)$ the angle of the vector $y_i-y_j$ relative to the 
first coordinate axis,
i.e. the phase of the complex number $w_i-w_j$. 
The energy of the trial state~\eqref{def:psi} is indeed best expressed in terms of $\Phit$, which is antisymmetric in its arguments if $\Phi$ is symmetric (respectively symmetric if $\Phi$ is antisymmetric). We shall prove the following (we write $ a \lesssim b$ when there is a constant $c>0$ depending only on $n$ such that $a\leq c b$):

\begin{theorem}[\textbf{Statistics transmutation in the ansatz~\eqref{def:psi}}]\label{thm:transmutation}\mbox{}\\
Fix an integer $n\ge 2$ (further constants depend on it). Let $\Psi_{\Phi}$ be as in~\eqref{def:psi}-\eqref{eq:choice Phi} with $\B=N \in \N$ and $\Phit \in C^0(\R^{2n})$ with support in $\DROPn$. Assume 

\medskip 

\noindent \emph{(i)} that there exists a constant $C_{\Phit} >0$ such that, for any $1\leq i\neq j \leq n$ 
\begin{equation}\label{eq:Phi on merging}
|\Phit (\by)| \leq C_{\Phit} |y_i-y_j|. 
\end{equation}

\medskip 

\noindent \emph{(ii)} that, for some (possibly very large) $s>0$ and all $j$
\begin{equation}\label{eq:stupid pre}
 \int_{\R^{2n}}  \left| \left(-\i \nabla_{y_j}   - (q-1)\B y_j^\perp \right) \Phit \right|^2 \lesssim N^s.
\end{equation}

\medskip

Then, for any $j \in \{1,\dots,n\}$ and for some $\kappa >0$ large enough, we have,  as $N\to+\infty$,
\begin{equation}\label{eq:main estimate gauge}
 \int_{\R^{2(N+n)}} \left| \left(-\i \nabla_{y_j} - q \B y_j^\perp\right) \Psi_\Phi \right|^2 = 2\B + \int_{\R^{2n}}  \left| \left(-\i \nabla_{y_j}   - (q-1)\B y_j^\perp \right) \Phit \right|^2 + \mathrm{Err} (\Phit)
 \end{equation}
 with an error term satisfying
 \begin{multline}\label{eq:main estimate gauge 2}
 \left| \mathrm{Err} (\Phit) \right| \lesssim C_{\Phit} ^2 \frac{\log N}{N} + C_{\Phit}  \sqrt{\frac{\log N}{N}} \left(\int_{ \DROPn\setminus \DROPn^{\varnothing}} \left|\left(-\i \nabla_{y_j}   - (q-1)\B y_j^\perp \right)  \Phit \right|^2\right) ^{1/2}.
\end{multline}
\end{theorem}

We shall illustrate the efficiency of our bound by discussing some applications in Section~\ref{sec:applications}. Before that,  some comments are in order:

\medskip 

\noindent\textbf{(1)} The above says that the kinetic energy of the tracer particles in the joint state~\eqref{def:psi} is well approximated by a free kinetic energy in the reduced state $\Phit$, whose quantum statistics is the opposite of the original one imposed on the tracers. The term $2\B$ in the right-hand side of~\eqref{eq:main estimate gauge} comes from an effective scalar potential experienced by the tracers, which turns out to be constant to a very good approximation. 

\medskip

\noindent\textbf{(2)} The occurence of Effect \boxed{\textbf{A}} is apparent in the modified magnetic field experienced by $\Phit$. The ``charge'' $-2q$ of tracer particles is effectively reduced to $-2(q-1)$, for they experience the constant density of magnetic flux attached to bath particles. This conclusion demands that the tracers be confined to the droplet of bath particles, whence our assumption on the support of $\Phi$. A sufficiently fast decay can replace the compact support, and can be enforced by including a trapping potential in the $W$ term of~\eqref{eq:intro hamil}.  
  
\medskip 

\noindent\textbf{(3)} Effect \boxed{\textbf{B}} is manifest in the modified statistics of $\Phit$. The effect occurs through the attachment of Aharonov-Bohm flux tubes to the tracers, as apparent when phrasing the theorem in the equivalent form
 \begin{multline}\label{eq:main estimate}
 \left|\int_{\R^{2(N+n)}} \left| \left(-\i \nabla_{y_j} - q \B y_j^\perp\right) \Psi_\Phi \right|^2 - 2\B - \int_{\R^{2n}}  \left| \left(-\i \nabla_{y_j}   + \bAt_j (\by) \right) \Phi \right|^2 \right| \lesssim C_{\Phit} ^2 \frac{\log N}{N} \\ + C_{\Phit}  \sqrt{\frac{\log N}{N}} \left(\int_{ \DROPn\setminus \DROPn^{\varnothing}} \left|\left(-\i \nabla_{y_j}   - (q-1)\B y_j^\perp \right)  \Phit \right|^2\right) ^{1/2}
\end{multline}
with $\bAt (\by)$ the total vector potential 
\begin{equation}\label{eq:gauge pot}
\bAt_j (\by) := - (q-1)\B y_j^\perp - \sum_{\ell \neq j} \frac{(y_j-y_\ell) ^{\perp}}{|y_j-y_\ell|^2}. 
\end{equation}
The second term in the above corresponds to an Aharonov-Bohm magnetic field 
$$ \curl_{y_j} \sum_{\ell \neq j} \frac{(y_j-y_\ell) ^{\perp}}{|y_j-y_\ell|^2} = 2 \pi \sum_{\ell \neq j} \delta_{y_\ell=y_j}.$$
It can be gauged away using the transformation~\eqref{eq:choice Phi}. Namely, a calculation shows that: 
\begin{equation} \label{eq:gauge}
\int_{\R^{2n}}  \left| \left(-\i \nabla_{y_j}   + \bAt (y_j) \right) \Phi \right|^2 = \int_{\R^{2n}}  \left| \left(-\i \nabla_{y_j}   - (q-1)\B y_j^\perp \right) \Phit \right|^2,
\end{equation}
which leads to~\eqref{eq:main estimate}. The coupling to Aharonov-Bohm flux tubes thus has the net effect of changing the tracers' quantum statistics. 
 
\medskip 

\noindent\textbf{(4)} This change of statististics occurs when tracer particles are well-separated on the scale of the magnetic length ${\B}^{-1/2}$. Our assumption~\eqref{eq:Phi on merging} allows to control the contribution from tracer particles getting very close. It is typically satisfied if $\Phit$ in~\eqref{eq:choice Phi} describes free fermions (i.e. tracers described by $\Phi$ were originally bosons) in the form of a Slater determinant. Otherwise one needs $W$ to incorporate some inter-particle repulsion in~\eqref{eq:intro hamil} to enforce~\eqref{eq:Phi on merging} in states of interest.  
 
\medskip

\noindent\textbf{(5)} The main sources of error in our estimates are due to the behavior of the system when tracer particles are close, i.e. in the vicinity of the diagonal set of the configuration space, and in fact \eqref{eq:Phi on merging}-\eqref{eq:stupid pre} can be replaced by local conditions around the diagonals. We will actually deduce the stated bound from a more general one, including smaller corrections to the effective emergent scalar and vector potentials due to binary tracer encounters. Scalar interactions can induce behavior akin to statistics~\cite{LarLunNam-21}, and thus it is necessary to obtain precise control of both emergent effects (vector and scalar) in order to isolate the genuine features of statistics transmutation.

\medskip
 
\noindent\textbf{(6)} We conjecture that our energy upper bound is optimal in the limit we consider, if the limit $g\to \infty$ (strong coupling) is also taken fast enough. Exactly how fast probably depends on a variant of the unsolved ``spectral gap conjecture'' in fractional quantum Hall physics~\cite{NacWarYou-20a,NacWarYou-20b,Rougerie-xedp19,RouSerYng-13b,WarYou-21} (see Section~\ref{sec:delta int} for more comments).  

\medskip

\noindent\textbf{(7)} A possible scenario leading to the model Hamiltonian we consider here, is to allow that a total number of $n+N$ fermionic particles distribute into different Landau levels, with the majority $N$ in the LLL and a small fraction $n \ll N$ in the higher Landau levels, i.e. excited from the ground state energy. 
This effectively defines several species of fermions, distinguished by their Landau level index, and what we have termed tracer particles are then simply the excited particles. In our theorem it is not necessary to assume that $q \neq 1$ or that $m$ is large, as done in adiabatic treatments. Rather it is only necessary that the fraction $n/N$ of excited particles is sufficiently small.

\medskip

\noindent\textbf{(8)} A conventional route to deducing statistics transmutation in quantum systems of the considered kind 
is to compute the corresponding gauge potential as the Berry connection
of the parametrized family of states $\Psiqh$, and then use an adiabatic argument to conclude that the tracer particles' motion around bath particles gives rise to enclosed fluxes,
or holonomies, thereby manifesting the appropriate exchange phases \cite{AroSchWil-84,Forte-92,KjoMyr-99,Myrheim-99,Jain-07,BonGurNay-11}. 
However, as discussed in~\cite{Forte-91} and \cite{LunRou-16}, this approach is lacking essential information on the dynamics and energy scales, such as those set by inherent as well as emergent interaction potentials, which necessitates knowledge of the full Hamiltonian. 
In addition, one must be careful to eliminate other potential ambiguities in the phase, cf. \cite[Section~9.8.2]{Jain-07}, \cite{KjoMyr-99}.
Furthermore the Born-Oppenheimer approximation may break down if we consider the bath and tracers to be the same type of particle. 
These issues are circumvented in our present approach.

\medskip

There is a natural generalization of Theorem~\ref{thm:transmutation} to the case where fractional statistics emerge. We state it as a conjecture, with motivation from the fractional quantum Hall effect literature~\cite{AroSchWil-84,Halperin-84,BonGurNay-11}.

\begin{conjecture}[\textbf{Emergence of fractional statistics}]\label{conj:fractional}\mbox{}\\
Fix an integer $n\ge 2$ (further constants depend on it). Let $\Psi_{\Phi}$ be as in~\eqref{def:psi}, but now with $p,\mu \in \N$ and
\begin{equation} \label{def:Psiqh frac}
\Psi_{\rm qh}(\w;\z) : = \left(\prod_{j=1}^n\prod_{k=1}^N (w_j-z_k)^p \right) \left(\prod_{1\leq k < \ell \leq N} (z_k-z_\ell)^\mu \prod_{k=1}^N e^{-\B |z_k|^2/2}\right),
\end{equation}
with $\B=N \in \N$ and $\Phi \in C^0(\R^{2n})$ with support in $\left\{ |x| \leq d \right\}^{\times n}$ for some $d< \sqrt{\mu}$. Assume the existence of a constant $C_{\Phi} >0$ such that 
\begin{equation} \label{def:Phi Hoelder}
|\Phi (\by)| \leq C_{\Phi} |y_i-y_j|^\beta,
\end{equation}
holds for some $\beta > 0$ and all $i \neq j$.

Then, for any $j \in \{1,\dots,n\}$ and for some $\kappa >0$ large enough, we have,  as $N\to+\infty$,
\begin{multline}\label{eq:main estimate frac}
 \int_{\R^{2(N+n)}} \left| \left(-\i \nabla_{y_j} - q \B y_j^\perp\right) \Psi_\Phi \right|^2 =  2\B\frac{p}{\mu} + \int_{\R^{2n}}  \left| \left(-\i \nabla_{y_j}   + \bAt_j (\by) \right) \Phi \right|^2  + \mathrm{Errors}
 \end{multline}
with $\bAt (\by)$ the total vector potential 
\begin{equation}\label{eq:gauge pot frac}
\bAt_j (\by) := - \left(q-\frac{p}{\mu}\right)\B y_j^\perp - \frac{p^2}{\mu}\sum_{\ell \neq j} \frac{(y_j-y_\ell) ^{\perp}}{|y_j-y_\ell|^2}. 
\end{equation}
\end{conjecture}

The precise form of the errors is not very important for our discussion, but should roughly look as in~\eqref{eq:main estimate} with powers depending on $\beta$. 
Formal calculations backing the above can be found in~\cite{LunRou-16}.  Essentially, one follows the heuristics we recall below, without having at disposal the determinantal structures which allow us to rigorously settle the case $p = \mu = 1$ in the present paper. Note that the second term in~\eqref{eq:gauge pot frac} can be \emph{formally} gauged away by choosing   
\begin{equation}\label{eq:choice Phi frac}
\Phi (y_1,\ldots,y_n) = \left(\prod_{1\leq i<j \leq n} e^{\i \frac{p^2}{\mu} \ang(y_i-y_j)} \right) \Phit (y_1,\ldots,y_n)
\end{equation}
and thus one thinks of the effective problem in terms of {\bf fractional statistics} (or, better, fractional exchange phases). Indeed if $\Phi$ is bosonic then $\Phit$ should now \emph{formally} pick a phase factor $e^{-\i \pi p^2 / \mu}$ upon particle exchange (respectively $-e^{-\i \pi p^2 / \mu}$ if $\Phi$ is fermionic). 
For $p>1$ this also accounts for the statistics of clusters of quasi-holes; cf. \cite{ThoWu-85}.
See again~\cite{LunRou-16} and references therein for further discussion.

\subsection{Outline of the proof}\label{sect:proof out}

The main ingredient~\eqref{def:Psiqh} in the construction of our trial state is also known as the ``Laughlin function of exponent $1$ with quasi-holes''. The second factor of the product is simply a (non-normalized) Slater determinant of the one-body orbitals 
\begin{equation}\label{eq:orbitals}
 \varphi_k(z) := \frac{\B^{(k+1)/2}}{\sqrt{\pi k!}} z^k e^{-\frac{\B}{2}|z|^2}
\end{equation}
for $k= 0,\ldots,N-1$. The compact form in~\eqref{def:Psiqh} follows from Vandermonde's formula
$$ \det_{N\times N} \left( \varphi_k (z_\ell) \right)_{k,\ell} \propto \prod_{1\leq k < \ell \leq N} (z_k-z_\ell).$$ 
This is a fermionic function which belongs to $\LLL^{\otimes N}$, as required in our model. The interpretation is that the state of the bath is made of one fermion per available orbital (here represented by angular momentum $k$ eigenfunctions~\eqref{eq:orbitals}), a logical choice for free particles
residing in the lowest Landau level.
Assuming a radial confinement of the bath, it is natural to minimize the angular momentum
which is equivalent to $|z|^2$ in the LLL.

The first factor of the product~\eqref{def:Psiqh} represents attachment of quasi-holes to the coordinates of the tracer particles: the joint wave function vanishes whenever $z_k=w_j$ for some pair $k,j$. This factor is an analytic function of $\z$, and symmetric under exchange of the labels of the $\z$ coordinates and of the $\w$ coordinates. Hence the full trial state~\eqref{def:psi} belongs to the state space~\eqref{eq:intro full space} respectively~\eqref{eq:intro full space bis}. The exact cancellation of the interspecies interaction comes at a price: analyticity in $\z$ must be preserved, which leads to phase singularities (vortices) attached to bath-tracer encounters. These are responsible for the emergent gauge field in Theorem~\ref{thm:transmutation}. We now sketch the main steps in the latter's proof. 

First, a direct but lengthy calculation shows that, for the basic ansatz~\eqref{def:psi}
\begin{equation}\label{eq:sketch exact}
\int_{\R^{2(N+n)}} \left| \left(-\i \nabla_{y_j} - q \B y_j^\perp\right) \Psi_\Phi\right|^2  =  \int_{\R^{2n}} \left| \left(-\i \nabla_{y_j} - q \B y_j^\perp + \mathcal{A}_j \right) \Phi\right|^2 + \int_{\R^{2n}} |\Phi|^2 \mathcal{V}_j
\end{equation}
for a vector potential 
$$ \cA_j = \cA_j (w_1,\ldots,w_n) \in \R^2$$
and a scalar potential 
$$ \cV_j = \cV_j (w_1,\ldots,w_n) \in \R$$
to be given in Section~\ref{sec:calculus}. These two potentials depend on the positions of the tracer particles and on some average properties of the state of the bath, conditioned on that of the tracers. The task is to show that, up to very small remainders, 
$$ \cA_j - q \B y_j^\perp \simeq \bAt_j (\by), \qquad \cV_j \simeq 2\B$$
with $\bAt (\by)$ as in~\eqref{eq:gauge pot}. The simplest expressions (but not the only useful ones) we can give are 
\begin{equation} \label{eq:vector pot}
\mathcal{A}_j (\w)  =  \frac{1}{2} \nabla^{\perp}_{w_j}   \log  \cqh^{-2}(\w).
\end{equation}
and
\begin{equation}\label{eq:scalar pot}
\mathcal{V}_j(\w)   = \frac{1}{2} \Delta_{w_j} \log \cqh^{-2}(\w)
\end{equation}
where $\cqh^{-2}(\w)$ is the normalization constant defined in~\eqref{def:cqh}. 

Second, we use the \emph{plasma analogy}~from \cite{Laughlin-83} (see also~\cite{LieRouYng-16,LieRouYng-17,RouSerYng-13b,Rougerie-xedp19,RouYng-14,RouYng-15,RouYng-17} for rigorous applications) to investigate the density of bath particles. In hopefully suggestive notation we rewrite
\begin{equation}\label{eq:plas ana}
\left|\cqh(\w) \Psi_{\rm qh}(\w;\z)\right|^2 = \frac{1}{\cZ (\w)} \exp\left( - T^{-1} H (\w;\z) \right),  
\end{equation}
interpreting  $\cqh(\w)^{-2} = \cZ (\w)$ as a \emph{partition function}, making the above right-hand side a probability measure for the positions $\z$. 
The (fictitious) plasma Hamiltonian is, for the general state \eqref{def:Psiqh frac},
\begin{equation}\label{eq:plas hamil}
H (\w;\z) = \sum_{j=1} ^N |z_j|^2 + \frac{2\mu}{\B} \sum_{1\leq i < j \leq N} \log \frac{1}{|z_i-z_j|} +  \frac{2p}{\B}  \sum_{j=1} ^N \sum_{k=1} ^n \log \frac{1}{|z_j-w_k|}. 
\end{equation}
One interprets the right-hand side of~\eqref{eq:plas ana} as the Gibbs state at temperature $T= \B^{-1} = N^{-1}$ of a 2D system of $N$  point charges (the $\z$ coordinates) repelling one another via 2D Coulomb forces, attracted to a constant background of opposite charge and repelled by $n$ other pinned point charges (the $\w$ coordinates). Screening properties of such a system suggest that its free energy $F(\w)$ behaves as the electrostatic energy of the $\w$ charges in the neutralizing background: 

\begin{heuristic}[\textbf{Screening in the fictitious plasma}]\label{heu:plasma}\mbox{}\\
There is a $\mathrm{constant}$ independent of $\w$ such that, for $\B = N \gg 1$,
\begin{align}\label{eq:free ener}
F(\w) &= - T \log \cZ (\w) = 2 \B^{-1} \log \cqh(\w) \nonumber\\
&\simeq -  \frac{p}{\mu}\sum_{j=1} ^n |w_j| ^2 - \frac{2p^2}{\mu\B} \sum_{1\leq i < j \leq n} \log \frac{1}{|w_i-w_j|} + \mathrm{constant}.
\end{align} 
\end{heuristic}

\begin{proof}[Argument]
We have the Gibbs variational principle
$$ 
F (\w)=  \min_{\nubf\in \mathcal{P}( \R^{2N})} \left\{\int_{\R^{2N}} H(\w;\z) \nubf (d\z) + T \int_{\R^{2N}} \nubf (\z) \log \left(\nubf (\z)\right) d\z\right\}
$$
where the minimum is over probability measures of the $\z$ coordinates. Neglecting the entropy term on the grounds that $T = N^{-1} \ll 1$ one finds 
\begin{equation}\label{eq:no entropy}
 F (\w) \simeq \min_{\z \in \R^{2N}} H(\w;\z).
\end{equation}
On the other hand, denoting 
\begin{equation}\label{eq:Coulomb}
 D \left( \varrho | \nu \right) = \iint_{\R^2\times \R^2}\varrho(x) \log\frac{1}{|x-y|} \nu(y) \,dx dy 
\end{equation}
the 2D Coulomb interaction between two charge densities, we can formally argue that 
\begin{equation}\label{eq:screening}
 H(\w;\z) \approx -  \frac{p}{\mu} \sum_{j=1} ^n |w_j| ^2 - \frac{2p^2}{\mu\B} \sum_{1\leq i < j \leq n} \log \frac{1}{|w_i-w_j|} + \frac{1}{\B} D\left( \rhot | \rhot \right) + \mbox{ constant}  
\end{equation}
where, with the choice $R = \sqrt{N\mu + np }$, 
\begin{equation}\label{eq:total charge}
\rhot := \sqrt{\mu} \sum_{j=1} ^N \delta_{z_j} + \frac{p}{\sqrt{\mu}} \sum_{j= 1} ^n \delta_{w_j} - \frac{1}{\pi \sqrt{\mu}} \1_{D(0,R)}
\end{equation}
is the total charge density of the fictitious plasma (in suitable units). Our convention (i.e. choice of the disk $D(0,R)$ of radius $R$) is that the system is neutral, 
\begin{equation}\label{eq:neutral}
 \int_{\R^2} \rhot  = 0. 
\end{equation}
We have used Newton's theorem to compute the potential generated by the background (last) term in~\eqref{eq:total charge}, and accepted two twists: (i) assuming that all the $\z$ and $\w$ charges lie within the disk $D(0,R)$ (ii) including in the constant term of~\eqref{eq:screening} several self-interactions like $D(\delta_{z_j}|\delta_{z_j})$. The latter are infinite, but they do not depend on the locations of the $\z$ or $\w$ particles. 

Using the neutrality condition~\eqref{eq:neutral} and Plancherel's formula we find, in Fourier space, 
$$ D\left( \rhot | \rhot \right) = \int_{\R^2} \frac{1}{2\pi |k|^2} \left| \widehat{\rhot} (k)\right|^2 dk$$
so that 
$$ 
D\left( \rhot | \rhot \right) \geq 0
$$
with equality if and only if $\rhot \equiv 0$ (see also~\cite[Chapter I, Lemma~1.8]{SafTot-97}). It follows that $\rhot \simeq 0$ for a minimizing configuration of the $\z$ charges. Inserting this in~\eqref{eq:no entropy}-\eqref{eq:screening}, we obtain the desired claim.
\end{proof}

Differentiating the right-hand side of~\eqref{eq:free ener} as in~\eqref{eq:vector pot}-\eqref{eq:scalar pot} yields the expected form of the main result~\eqref{eq:main estimate} (up to a singular term in the scalar potential that we discard for $w_i \neq w_j$),
respectively the conjectured relation~\eqref{eq:main estimate frac}.

The main difficulty to push this analogy to completion and provide a proof of our main results is that the precision of the screening Heuristic~\ref{heu:plasma} is not quite obvious, while we need to obtain extremely small remainders, and to be able to differentiate the results with respect to the tracers' location. The standard statistical mechanics~\cite{RouSerYng-13b} route analyzes a mean-field/small temperature limit for the classical plasma, and yields way too loose estimates. Even recent very refined analysis~\cite{ArmSer-20,BauBouNikYau-15,BauBouNikYau-16,LebSer-15,LebSer-16,Leble-15b,Serfaty-20} of such systems does not seem to provide what we would need. 

For these reasons we follow a different route, based on the well-known determinantal structure of the probability measure~\eqref{eq:plas ana} in the case $p=\mu=1$. 
Up to now the discussion was rather general and
applicable to support Conjecture~\ref{conj:fractional},
but we shall now heavily rely on special structures due to the state of the bath being a Slater (Vandermonde) determinant. The reach of the following discussion is thus limited to support Theorem~\ref{thm:transmutation}.

Hereafter we denote 
\begin{equation} \label{def:Psiqh void}
\Psi_{\rm qh}\left(\emptyset;\z\right) : =  \prod_{1\leq k < \ell \leq N} (z_k-z_\ell) \prod_{k=1}^N e^{-\B |z_k|^2/2}
\end{equation}
the wave-function of the bath in the absence of tracers and
\begin{equation} \label{def:cqh void}
 \cqh(\emptyset)^{-2} := \int_{\C^N} \left| \Psi_{\rm qh}(\emptyset ; \z)\right|^2 \d\z 
\end{equation}
the corresponding normalization constant. It is well-known that $\cqh(\emptyset)^{2} |\Psi_{\rm qh}(\emptyset;\z)|^2$ is the joint probability distribution of the eigenvalues of a random matrix drawn from  $\P_N$, the Ginibre ensemble~\cite{Ginibre-65}. A Ginibre random matrix is a $N\times N$ matrix filled with independent (complex) Gaussian random variables of variance $\B^{-1}$. 
Since such a matrix is non--Hermitian, its eigenvalues $\{\lambda_k\}_{k=1}^N$ form a point process\footnote{All the eigenvalues are distinct almost surely.} in $\C$. The latter is determinantal with a correlation kernel $\K_N$  which is given for any $\B>0$ and $N\in\N$ by 
\begin{equation} \label{def:KN}
\K_N(z,w) : = \sum_{0 \le j <N} \frac{\B^{j+1}}{\pi j!} z^j \overline{w}^j e^{-\B(|z|^2+|w|^2)/2}  , \qquad z,w\in\C .
\end{equation}
We review the properties of these kernels in Section~\ref{sect:kernel}. 
For now we point out that the (normalized) density of eigenvalues is $\rho(z) = N^{-1} \K_N(z,z)$ for any $z\in\C$. In particular, if we scale the variance so that $\B =N$, we have as $N\to+\infty$, 
\[
 \rho_N(z) \to \frac{1}{\pi} \1_{|z| < 1} \qquad\text{for almost all}\quad z\in\C,
\] 
which is known as the circular law~\cite{Girko-84,BorCha-12} and means that in the (macroscopic) regime where $\B =N$, the eigenvalues distribute uniformly in the unit disk $\D =\{ |x| <1\}$. This gives the equilibrium measure of the plasma discussed above (for $\w=\emptyset$) and its support $\D$ is conventionally called the \emph{droplet}.

We now denote 
\begin{equation}\label{eq:char pol}
Q_N(w) = \prod_{k=1}^N (w-\lambda_k) 
\end{equation}
the characteristic polynomial of a Ginibre random matrix.  Then, for any $n\in\N$ and $\w\in\C^n$, we have 
\begin{equation} \label{charpoly}
\E_N\left[ {\textstyle\prod_{j=1}^n} |Q_N(w_j)|^2 \right] =   \cqh(\emptyset)^{2} \cqh(\w)^{-2}
\end{equation}
which provides a random matrix interpretation of the  partition function $ \cqh(\w)$ from \eqref{def:cqh}.  This connection with the characteristic polynomial of the Ginibre ensemble motivated the proof of the following result, taken from~\cite{AkeVer-03,Lambert-20}:

\begin{theorem}[\textbf{Exact expression for partition functions with quasi-holes}] \label{thm:exp}\mbox{}\\
For any  $\w \in\C^n$ with $w_1\neq \dots \neq w_n$, we have 
\begin{align}\label{eq:exact part}
\cqh(\w)^{-2} &= \int_{\C^N} \left| \Psi_{\rm qh}(\w, \z)\right|^2\d\z\nonumber \\
&=  N! \pi^{N+n} {\textstyle  \prod_{k=1}^{N+n-1} k! } \, \B^{-(N+n)(N+n+1)/2}  \det_{n\times n} \left[\K_{N+n}(w_i,w_j) \right]  \frac{\prod_{j=1}^n e^{\B|w_j|^2}}{\left| \triangle(\w)\right|^{2}} , 
\end{align}
where 
\begin{equation}\label{eq:Vandermonde}
\triangle(\w) := \prod_{1\le i<j \le n}(w_j-w_i) 
\end{equation}
is a Vandermonde  determinant. 
\end{theorem}

The proof is a consequence of the determinantal structure of the Ginibre ensemble; we recall that in Appendix~\ref{Appendix}. Observe now that, taking the $\log$ of~\eqref{eq:exact part} and differentiating as in~\eqref{eq:vector pot}-\eqref{eq:scalar pot} yields exactly the expected expressions for the emerging potentials, up to derivatives of the extra term 
\begin{equation}\label{eq:error}
 \log \left( \det_{n\times n} \left[\K_{N+n}(w_i,w_j) \right]\right). 
\end{equation}
We shall use detailed asymptotics for the correlation kernel $\K_{N+n}$, presented in Section~\ref{sect:Kasymp}, as our main tool to complete the proof of Theorem~\ref{thm:transmutation} based on the exact formulas~\eqref{eq:vector pot}-\eqref{eq:scalar pot}-\eqref{eq:exact part}. 
This is done in Section~\ref{sec:estimates}, while some applications of our main result are discussed in Section~\ref{sec:applications}.

\bigskip

\noindent \textbf{Acknowledgments:} Funding from the European Research Council (ERC) under the European Union's Horizon 2020 Research and Innovation Programme (Grant agreement CORFRONMAT No 758620) is gratefully acknowledged (N.R).
D.L. gratefully acknowledges financial support from the G\"oran Gustafsson Foundation (grant no. 1804) and the Swedish Research Council (grant no. 2021-05328, ``Mathematics of anyons and intermediate quantum statistics'').
G.L. research is supported by the SNSF Ambizione grant S-71114-05-01. This project was started at Institut Mittag-Leffler, on the occasion of the July 2019 workshop ``New directions in mathematics of Coulomb gases and quantum Hall effect'', whose organisers (Robert Berman, Gaetan Borot, Semyon Klevtsov, Sylvia Serfaty, Paul Wiegmann) we warmly thank.

\section{Calculation of the emerging potentials}\label{sec:calculus}

\subsection{Notation}\label{sec:not}

Recall that in the sequel of this article, we adopt the scaling $\B =N$. We have seen that it is also convenient to switch between planar and complex variables; we identify the coordinates $\by=(y_1,\ldots, y_n) \in \R^{2n}$ of the tracer particles with complex numbers  $\w=(w_1, \dots , w_n) \in\C^n$ and the coordinates $\bx=(x_1,\ldots, x_n) \in \R^{2N}$ with $\z=(z_1,\dots, z_N)\in\C^N$. With these conventions, $u(\by;\bx) = u(\w;\z)$ is the same function, without implying that $u$ does not depend on 
the complex conjugates $(\overline{\bw},\overline{\bz})$. 
If $z = x+\i y \in \C$, we denote 
\[
\partial_z = \frac{\partial_x - \i \partial_y}{2} , \qquad \partial_{\overline z} = \frac{\partial_x + \i \partial_y}{2}
\] 
so that 
\begin{equation}\label{eq:Lap}
\Delta_z = 4 \partial_z \partial_{\overline z} = 4\partial_{\overline z} \partial_z. 
\end{equation}
Moreover, we have 
\begin{equation}\label{eq:grad 1}
\nabla = \begin{pmatrix} \partial_x \\ \partial_y\end{pmatrix} 
= \begin{pmatrix} \partial_z + \partial_{\overline z} \\  \i(\partial_z-  \partial_{\overline z}) \end{pmatrix}
\qquad\text{and} \qquad 
\nabla^\perp = \begin{pmatrix} - \partial_y \\ \partial_x \end{pmatrix} 
= \begin{pmatrix} \i(\partial_{\overline z}-\partial_z) \\  \partial_z +  \partial_{\overline z} \end{pmatrix}
\end{equation}
Hence, if  $\phi : \C \to \R$ is smooth and $\psi : \C \to \C$ is analytic, then 
\begin{equation}\label{eq:grad 2}
\nabla \phi = 2 \begin{pmatrix} \Re \partial_z\phi  \\ -\Im \partial_z\phi  \end{pmatrix} 
= 2 \begin{pmatrix} \Re \partial_{\overline z}\phi  \\ \Im \partial_{\overline z}\phi  \end{pmatrix}
\qquad\text{and} \qquad 
\nabla \psi = \begin{pmatrix} 1 \\ \i \end{pmatrix}\partial_{z} \psi
\end{equation}
since $\partial_{\overline z} \psi =0$. Another consequence is that if  $\phi : \C \to \R$ is smooth, then
\begin{equation}\label{eq:grad 3}
 |\nabla \phi| ^2 = 4 |\partial_{z} \phi|^2. 
\end{equation}

\subsection{Expressions using the partition function}

Our starting point is the next proposition, which applies equally well to the fractional statistics case~\eqref{def:Psiqh frac}.

\begin{proposition}[\textbf{Expressions of the emerging potentials}]\label{prop:exp}\mbox{}\\
Fix $n,N=\B\in \N$ and $q\in\R$. Let the joint wave-function $\Psi_\Phi$ be as in \eqref{def:psi}, with a smooth $\Phi$. We have for any $j=1,\dots, n$, 
\begin{equation}\label{eq:emerg exact}
\int_{\R^{2(N+n)}} \left| \left(-\i \nabla_{y_j} - q \B y_j^\perp \right) \Psi_\Phi\right|^2 d\bx \,d\by
 =  \int_{\R^{2n}}   \left| \left(-\i \nabla_{y_j} + \mathcal{A}_j - q \B y_j^\perp \right) \Phi \right|^2 d\by 
 + \int_{\R^{2n}} |\Phi|^2 \mathcal{V}_j 
\end{equation}
where
\begin{align} \label{eq:vector pot bis}
\mathcal{A}_j (\w)  &= \cqh^2(\w) \Im \left( \int_{\R^{2N}} \overline{\Psi_{\rm qh} (\w;\z)}  \nabla_{y_j}  \Psi_{\rm qh} (\w;\z) \d\z \right)
\nonumber\\
&=  \frac{1}{2} \nabla^{\perp}_{w_j}   \log  \cqh^{-2}(\w).
\end{align}
and
\begin{align}\label{eq:scalar pot bis}
\mathcal{V}_j(\w)  &= \cqh^2(\w)  \int_{\R^{2N}}  \left|  \nabla_{y_j}  \Psi_{\rm qh} (\w;\z) \right|^2 
\d\z - \cqh^4(\w)   \left| \int_{\R^{2N}}  \overline{\Psi_{\rm qh} (\w;\z)} \nabla_{y_j}  \Psi_{\rm qh} (\w;\z) \d\z \right|^2 \nonumber\\
&= \frac{1}{2} \Delta_{w_j} \log \cqh^{-2}(\w)
\end{align}
where $\cqh^{-2}(\w)$ is the normalization constant defined in~\eqref{def:cqh}. 
\end{proposition}

The second expressions in~\eqref{eq:vector pot bis}-\eqref{eq:scalar pot bis} are those we announced already. The first ones will be useful as well. In particular, applying Cauchy-Schwarz to the first line of~\eqref{eq:scalar pot bis} we immediately have $\cV_j \geq 0$.

\begin{proof}
By symmetry, it suffices to treat the case $j=1$.  

\noindent\textbf{First expressions for the emerging potentials.} Some of the calculations below might be folklore in adiabatic theory. It is convenient to denote 
$$ \Xi (\by;\bx) = \cqh(\w) \Psi_{\rm qh}(\w;\z)$$
so that 
$$ \Psi_\Phi (\by;\bx) = \Phi (\by)\Xi (\by;\bx).$$
We first observe that 
\begin{equation}\label{eq:use norm}
\Re \left\{\int_{\R^{2N}} \Xi (\by;\bx) \overline{ \nabla_{y_1} \Xi (\by;\bx) } d\bx \right\} = 0
\end{equation}
since by definition 
$$\int_{\R^{2N}} |\Xi (\by;\bx)|^2 d\bx \equiv 1.$$
Then 
\begin{align}\label{eq:calc vec}
\left| \left(-\i \nabla_{y_1} - q \B y_1^\perp\right) \Psi_\Phi\right| ^ 2 &= \left| \Phi (-\i \nabla_{y_1} \Xi) + \Xi \left( -\i \nabla_{y_1} \Phi - q \B y_1 ^{\perp} \Phi \right)\right|^2 \nonumber \\
&=|\Xi| ^2 \left| \left(-\i \nabla_{y_1} - q \B y_1^\perp\right) \Phi \right| ^ 2 + |\Phi|^2 |\nabla_{y_1} \Xi|^2 \nonumber \\
&+ \Xi \overline{\Phi} \left(-\i \nabla_{y_1} - q \B y_1^\perp\right) \Phi \cdot \i \overline{\nabla_{y_1} \Xi} + \overline{\Xi} \Phi \overline{\left(-\i \nabla_{y_1} - q \B y_1^\perp\right) \Phi} \cdot (-\i \nabla_{y_1} \Xi).
\end{align}
With 
$$ 
\cA_1 (\by) := \i \int_{\R^{2N}} \overline{\Xi} \nabla_{y_1} \Xi d\bx 
$$
it follows from~\eqref{eq:use norm} that $\cA_1$ is real, and since so is $\cqh(\w)$, this definition coincides with the first expression in~\eqref{eq:vector pot bis}. Integrating~\eqref{eq:calc vec} with respect to $\bx$ and completing a square we find   
$$
\int_{\R^{2(N)}} \left| \left(-\i \nabla_{y_1} - q \B y_1^\perp \right) \Psi_\Phi\right|^2 
 =  \left| \left(-\i \nabla_{y_1} + \mathcal{A}_1 - q \B y_1^\perp \right) \Phi \right|^2   
 + |\Phi|^2 \left( \int_{\R^{2N}}|\nabla_{y_1} \Xi| ^2 - |\cA_1| ^2\right)
$$
which leads to~\eqref{eq:emerg exact} with 
$$ 
\cV_1 = \int_{\R^{2N}}|\nabla_{y_1} \Xi| ^2 - |\cA_1| ^2.
$$
To complete the proof, we note that~\eqref{eq:use norm} leads to 
\begin{equation}\label{eq:use norm 2} 
\cqh \Re \left( \int_{\R^{2N}} \overline{\Psi_{\rm qh}} \nabla_{y_1} \Psi_{\rm qh} \right) = - \left( \nabla_{y_1} \cqh\right) \int_{\R^{2N}} \left| \Psi_{\rm qh}\right|^2
\end{equation}
and hence
$$ \int_{\R^{2N}} |\nabla_{y_1} \Xi| ^2 =  \cqh^2 \int_{\R^{2N}} \left|\nabla_{y_1} \Psi_{\rm qh}\right|^2 - \left| \nabla_{y_1} \cqh \right|^2 \int_{\R^{2N}} \left| \Psi_{\rm qh}\right|^2.$$
On the other hand 
$$ |\cA_1| ^2 =  \cqh^4 \left| \int_{\R^{2N}} \overline{\Psi_{\rm qh}} \nabla_{y_1} \Psi_{\rm qh} \right|^2 -  \cqh^4 \Re \left( \int_{\R^{2N}} \overline{\Psi_{\rm qh} }  \nabla_{y_1}  \Psi_{\rm qh} \right)^2$$
so that, using~\eqref{eq:use norm 2} again and recalling that 
$$  \cqh^2 \int_{\R^{2N}} \left| \Psi_{\rm qh} \right|^2 \equiv 1$$
yields~\eqref{eq:emerg exact} with the expressions in the first lines of~\eqref{eq:vector pot bis} and~\eqref{eq:scalar pot bis}.

\medskip 

\noindent\textbf{Second expressions for the emerging potentials.} Since the function $w_1 \mapsto  \Psi_{\rm qh} (\w;\z)$ is analytic, we have 
\begin{equation}\label{eq:use analytic}
\nabla_{y_1}  \Psi_{\rm qh}  = { \begin{pmatrix} 1 \\ \i \end{pmatrix}}  \partial_{w_1}  \Psi _{\rm qh}  
\end{equation}
so that
\begin{equation} \label{A1}
\begin{aligned}
\cqh ^2  \Im \left( \int_{\R^{2N}} \overline{\Psiqh}  \nabla_{y_1}  \Psiqh  \right) 
& =  \cqh ^2  \Im \bigg(   \begin{pmatrix} 1 \\ \i \end{pmatrix} \partial_{w_1}   \int_{\C^N} \big|\Psiqh (\w;\z) \big|^2\d\z \bigg)  \\
& = \Im \Big(  {\footnotesize	 \begin{pmatrix} 1 \\ \i \end{pmatrix}} \partial_{w_1}   \log  \cqh^{-2}(\w)   \Big)  . 
\end{aligned}
\end{equation}
Where we used the first identity in~\eqref{eq:grad 2}, valid for real-valued functions. The second expression in~\eqref{eq:vector pot bis} follows. 

As regards the scalar potential, we use~\eqref{eq:grad 3} on real-valued functions and~\eqref{eq:Lap}. Combining with~\eqref{eq:use analytic} and~\eqref{def:cqh} again we find
\begin{equation*}
\begin{aligned} 
\mathcal{V}_1(\w)  & =\cqh^{2}(\w)   \int_{\C^N}  \big|  \nabla_{w_1}  \Psiqh(\w;\z)\big|^2 \d\z
- \cqh^{4}(\w)   \bigg| \int_{\C^N}  \overline{\Psiqh(\w;\z)} \nabla_{w_1}  \Psiqh(\w;\z)\d\z \bigg|^2  \\
& =  \frac{1}{2} \cqh^{2}(\w) \Delta_{w_1}\bigg(   \int_{\C^N}  \big| \Psiqh(\w;\z)\big|^2 \d\z \bigg)
- \frac{1}{2} \cqh^{4}(\w) \bigg|  \nabla_{w_1}   \int_{\C^N} \big|\Psiqh(\w;\z) \big|^2\d\z \bigg|^2 \\
& 
 = \frac{1}{2} \Delta_{w_1} \log \cqh^{-2}(\w).
\end{aligned}
\end{equation*}
This vindicates the second line of~\eqref{eq:scalar pot bis}.
\end{proof}

\subsection{Integral expressions}

We finally give a third expression of both potentials, less transparent but useful for some estimates below. 

\begin{notation}[\textbf{Remainders in the partition function}]\label{not:remain}\mbox{}\\
For any $n\in\N$ and $\w \in\C^n$, we denote  for $N\in \{1, \dots, \infty\}$, 
 \begin{equation} \label{def:Upsilon}
\Upsilon_N(\w) := \det \big[\pi \B^{-1} \K_{N+n}(w_i,w_j) \big]_{n \times n} 
\end{equation}
where $K_{N+n}$ is as in~\eqref{def:KN}. Theorem~\ref{thm:exp} can then be rephrased as 
\begin{equation}\label{cqh}
 \cqh^{-2} (\w) = \Gamma_N^n  \frac{\prod_{j=1}^n e^{\B|w_j|^2}}{\left| \triangle(\w)\right|^{2}} \Upsilon_N(\w) ,
 \qquad  \Gamma_N^n := N! \B^{-\frac{(N+n)(N+n-1)}{2} + N } \pi ^{N} {\textstyle   \prod_{k=1}^{N+n-1} k! }
\end{equation}
for all $\w \in \C^n$.
\end{notation}

Observe that since $K_{N+n}$ is a correlation kernel (see Section~\ref{sect:kernel} below), the matrix on the right-hand side of \eqref{def:Upsilon} is positive definite if and only if $w_1\neq \cdots \neq w_n$, 
and the normalization is such that for any $N,n\in\N$ and all $\w \in\C^n$, 
\begin{equation} \label{Upsilon1}
\Upsilon_N(\w)  \in [0,1] . 
\end{equation}
Indeed, by  Hadamard's inequality and the fact that $\K_{N+n}(w,w) \le \B/\pi$ for all $w\in\C$ (see Equation~\eqref{Kasymp2}), we have $\Upsilon_N(\w) \le 1$ for all $\w\in\C^n$. We use the function $\Upsilon_N$ to obtain integral formulae for the emerging potentials using complex coordinates.

\begin{proposition}[\textbf{Integral formulae for the emerging potentials}]\label{prop:int}\mbox{}\\
Fix $n,N \in\N$ and $\w\in \C^n$. With $\Upsilon=\Upsilon_N$, it holds for $j=1,\dots,n$, 
\begin{equation}\label{eq:vect int}
\mathcal{A}_j (\w)  =  N  \Im\bigg(  {\footnotesize	 \begin{pmatrix} 1 \\ \i \end{pmatrix}} \int_\C  \frac{\Upsilon(\w,z)}{\Upsilon(\w)} \frac{\d^2z}{\pi(w_j-z)} \bigg)
\end{equation}
and 
\begin{equation}\label{eq:scal int}
\mathcal{V}_j(\w) =  \frac{2N}{\Upsilon(\w)} \bigg( \int_\C  \frac{\Upsilon(\w,z)}{|w_j-z|^2} \frac{\d^2z}{\pi} 
  + \frac{N}{\Upsilon(\w)}  \iint_{\C^2} \frac{\Upsilon(\w,z,\zeta)\Upsilon(\w)- \Upsilon(\w,z) \Upsilon(\w,\zeta) }{ (w_j-z)(\overline{w_j-\zeta})} \frac{\d^2z}{\pi} \frac{\d^2\zeta}{\pi} \bigg) . 
\end{equation}
\end{proposition}

\begin{proof}[Proof of~\eqref{eq:vect int}]
We start from the first expression in~\eqref{eq:vector pot bis} and the definition~\eqref{def:Psiqh}. Since the function $\Psi_{\rm qh}(\w;\z)$ is analytic in $w_1 \in\C$, we have
\[\begin{aligned}
\nabla_{w_1}\Psi_{\rm qh}(\w;\z)  
&= \begin{pmatrix} 1 \\ \i \end{pmatrix} \partial_{w_1}\Psi_{\rm qh}(\w;\z)  \\
& = \begin{pmatrix} 1 \\ \i \end{pmatrix} \sum_{\ell=1}^N  \frac{1}{w_1-z_\ell}  \prod_{k=1}^N e^{-\B |z_k|^2/2}\prod_{j=1}^n (w_j-z_k)  \prod_{q=k+1}^N (z_q-z_k)  \\
&= \begin{pmatrix} 1 \\ \i \end{pmatrix} 
 \sum_{\ell=1}^N  (-1)^{N-\ell} e^{-\B |z_\ell|^2/2}  \prod_{j=2}^n (w_j-z_\ell) 
\prod_{\substack{k=1 \\ k\neq \ell}}^N e^{-\B |z_k|^2/2}\prod_{\substack{j=1 \\ w_{n+1} = z_\ell}}^{n+1} (w_j-z_k)  \prod_{\substack{q=k+1\\ q\neq \ell}}^N (z_q-z_k) . 
\end{aligned}\]
This shows that 
\begin{equation} \label{def:Psi'}
\nabla_{w_1}\Psi_{\rm qh}(\w;\z) =   \begin{pmatrix} 1 \\ \i \end{pmatrix} 
 \sum_{\ell=1}^N  (-1)^{N-\ell}  e^{-\B |z_\ell|^2/2} \prod_{j=2}^n (w_j-z_\ell)  \Psi_{\rm qh}(\w_{(\ell)};\z^{(\ell)}) 
\end{equation}
where $\w_{(\ell)} = (w_1,\dots, w_{n},z_\ell)$ and $\z^{(\ell)} = \z \setminus\{z_\ell\}$. 
Similarly, we check that for any $\ell=1,\dots, N$, 
\[
\Psi_{\rm qh}(\w;\z) =  (-1)^{N-\ell}   e^{-\B |z_\ell|^2/2} \prod_{j=1}^n (w_j-z_\ell) 
 \Psi_{\rm qh}(\w_{(\ell)};\z^{(\ell)})  . 
  \]
Hence, we deduce from the previous formulae that for any $\w\in\C^n$ with $w_1\neq \cdots \neq w_n$,  
\begin{equation*} 
\int_{\C^{N}}  \overline{\Psi_{\rm qh}(\w;\z)  } \nabla_{w_1}  \Psi_{\rm qh}(\w;\z)   \d\z  
=  \begin{pmatrix} 1 \\ \i \end{pmatrix}  N  \int_\C \frac{e^{-\B |z|^2}}{w_1-z}  \prod_{j=1}^n |w_j-z|^2\int_{\C^{N-1}} \big| \Psi_{\rm qh}(\w';\z')\big|^2 \d\z' \d z
\end{equation*}
where we used the symmetry with respect to $\z \in\C^N$ and we let  $\w' = (w_1,\dots, w_{n},z)$, $\z' = (z_1,\dots, z_{N-1})$.

\medskip

Moreover, using Theorem~\ref{thm:exp} formulated as in~\eqref{cqh}, we obtain the ratio
 \begin{align} \label{exp0}
  \cqh^2(\w) \int_{\C^{N-1}} \big| \Psi_{\rm qh}(\w';\z')\big|^2 \d\z'  
& =  \frac{\Gamma_{N-1}^{n+1}}{\Gamma_{N}^n} \left|  \frac{\Delta(\w)}{\Delta(\w')} \right|^2  \frac{\Upsilon(\w')}{\Upsilon(\w)} e^{\B|z|^2}
\end{align}
where $\Upsilon$ is as in~\eqref{def:Upsilon} and  $\frac{\Gamma_{N-1}^{n+1}}{\Gamma_{N}^n} = 1/\pi$.  This implies that 
\begin{equation*} 
 \cqh^2(\w) \int_{\C^{N}}  \overline{\Psi_{\rm qh}(\w;\z)  } \nabla_{w_1}  \Psi_{\rm qh}(\w;\z)   \d\z  
=  \begin{pmatrix} 1 \\ \i \end{pmatrix}  N  \int_\C \frac{1}{w_1-z} \frac{\Upsilon(\w')}{\Upsilon(\w)} \frac{\d z}{\pi}
\end{equation*}
By \eqref{eq:vector pot bis}, this completes the proof of the integral formula \eqref{eq:vect int}.
\end{proof}

\begin{proof}[Proof of~\eqref{eq:scal int}]
According to \eqref{eq:scalar pot bis}, there only remains to show that
\begin{equation}\label{eq:scal int bis}
 \cqh^2(\w) \int_{\C^N}  \big|\nabla_{w_1}\Psi_{\rm qh}(\w;\z) \big|^2 \d\z
=  \frac{2N}{\Upsilon(\w)} \bigg( \int_\C  \frac{\Upsilon(\w,z)}{|w_1-z|^2} \frac{\d^2z}{\pi} 
  + N  \iint_{\C^2} \frac{\Upsilon(\w,z,\zeta)}{ (w_1-z)(\overline{w_1-\zeta})} \frac{\d^2z}{\pi} \frac{\d^2\zeta}{\pi} \bigg). 
 \end{equation}
 
Using formula \eqref{def:Psi'}, we obtain
\[
 \big|\nabla_{w_1}\Psi_{\rm qh}(\w;\z) \big|^2 
= 2 \bigg|  \sum_{\ell=1}^N e^{-\B |z_\ell|^2/2}  (-1)^{\ell}  \prod_{j=2}^n (w_j-z_\ell) \Psi_{\rm qh}(\w_{(\ell)};\z^{(\ell)}) \bigg|^2
\]
where $\w_{(\ell)} = (w_1,\dots, w_{n},z_\ell)$ and $\z^{(\ell)} = \z \setminus\{z_\ell\}$. 
Expanding the square, this leads to
\begin{multline*}
 \big|\nabla_{w_1}\Psi_{\rm qh}(\w;\z) \big|^2 =
\\ 4 \sum_{1\le\ell<k\le N}  e^{-\B(|z_\ell|^2+|z_k|^2)/2} (-1)^{k-\ell}  
\Re\bigg\{ \prod_{j=2}^n (w_j-z_\ell) \overline{(w_j-z_k)} \Psi_{\rm qh}(\w_{(\ell)};\z^{(\ell)}) \overline{\Psi_{\rm qh}(\w_{(k)};\z^{(k)})}  \bigg\} \\
\qquad + 2 \sum_{\ell=1}^N e^{-\B |z_\ell|^2}  \prod_{j=2}^n |w_j-z_\ell|^2 \big| \Psi_{\rm qh}(\w_{(\ell)};\z^{(\ell)}) \big|^2 . 
\end{multline*}
One may also check that for any $k=1,\dots, N$ with $k\neq\ell$,  
\[
\Psi_{\rm qh}(\w_{(\ell)};\z^{(\ell)})= (-1)^{N-\ell} e^{-\B |z_k|^2/2}(z_\ell-z_k) \prod_{j=1}^{n} (w_j-z_r) 
  \Psi_{\rm qh}(\w_{(\ell,k)};\z^{(\ell,k)})  , 
  \]
where $\w_{(\ell,k)}= (w_1,\dots, w_{n},z_\ell, z_k)$ and $\z^{(\ell,k)} = \z \setminus\{z_\ell, z_k\}$. Using this, we obtain
\begin{multline*}
 \big|\nabla_{w_1}\Psi_{\rm qh}(\w;\z) \big|^2 
  =\\  4 \sum_{1\le\ell<k\le N}   e^{-\B(|z_\ell|^2+|z_k|^2)} \Re\bigg\{ \frac{|z_\ell-z_k|^2}{ (w_1-z_\ell)(\overline{w_1-z_k})}\bigg\}\prod_{j=1}^n |w_j-z_\ell|^2 |w_j-z_k|^2 \big|  \Psi_{\rm qh}(\w_{(\ell,k)};\z^{(\ell,k)}) \big|^2 \\
  \qquad + 2 \sum_{\ell=1}^N e^{-\B |z_\ell|^2}  \prod_{j=2}^n |w_j-z_\ell|^2 \big| \Psi_{\rm qh}(\w_{(\ell)};\z^{(\ell)}) \big|^2 . 
\end{multline*}
Hence, by symmetry with respect to permutations in the variable $\z$, this implies that 
 \begin{multline} \label{exp1}
\int_{\C^N}  \big|\nabla_{w_1}\Psi_{\rm qh}(\w;\z) \big|^2 \d\z
\\=  2 \pi \int_\C  \prod_{j=2}^n |w_j-z|^2 \bigg( \int_{\C^{N-1}}\big| \Psi_{\rm qh}(\w';\z') \big|^2 \d\z'  \bigg) \phi_\B(\d z) +  2\pi^2 \frac{N-1}{N} \\
 \iint_{\C^2}\Re\bigg\{ \frac{|z-\zeta|^2}{ (w_1-z)(\overline{w_1-\zeta})}\bigg\}\prod_{j=1}^n |w_j-z|^2 |w_j-\zeta|^2
\bigg( \int_{\C^{N-2}} \big|  \Psi_{\rm qh}(\w'';\z'') \big|^2  \d\z'' \bigg) \phi_\B(\d z)\phi_\B(\d\zeta), 
\end{multline}
where $\w'=(w_1,\dots, w_n,z)$, $\w''=(w_1,\dots, w_n,z,\zeta)$, $\z'=(z_1,\dots, z_{N-1})$, $\z''=(z_1,\dots, z_{N-2})$ and 
$$\phi_{\B} = \frac{ \B e ^{-\B|x|^2}}{\pi}.$$ 
Then, using~\eqref{cqh},  we obtain for  $\w\in\C^n$, 
 \[
 \cqh^2(\w) \int_{\C^{N-2}} \big|  \Psi_{\rm qh}(\w'';\z'') \big|^2  \d\z''   =  \frac{\Gamma_{N-2}^{n+2}}{\Gamma_{N}^n} \frac{\Upsilon(\w'')}{\Upsilon(\w)}
 \frac{| \triangle(\w)|^2}{| \triangle(\w'')|^2}  e^{\B(|z|^2+|\zeta|^2)} .
\]
Since $ \frac{| \triangle(\w)|^2}{| \triangle(\w'')|^2} |z-\zeta|^2\prod_{j=1}^n |w_j-\zeta|^2 |w_j-z|^2 =1$ and $\frac{\Gamma_{N-2}^{n+2}}{\Gamma_{N}^n} =  \frac{N/\pi^2}{N-1}$, by combining this formula with \eqref{exp0} and \eqref{exp1}, we conclude that
\begin{multline*}
\int_{\C^N}  \big|\nabla_{w_1}\Psi_{\rm qh}(\w;\z) \big|^2 \d\z
= 2N \int_\C \frac{1}{|w_1-z|^2} \frac{\Upsilon(\w')}{\Upsilon(\w)}  \frac{\d z}{\pi} \\
+ 2 N^2  \iint_{\C^2}\Re\bigg\{ \frac{\Upsilon(\w'')/\Upsilon(\w)}{ (w_1-z)(\overline{w_1-\zeta})}\bigg\}  \frac{\d^2z}{\pi} \frac{\d^2\zeta}{\pi} . 
\end{multline*}
This completes the proof of formula \eqref{eq:scal int bis} since the second integral on the right-hand side of \eqref{eq:scal int bis} is real-valued (by symmetry $z \leftrightarrow \zeta$). 
\end{proof}

\section{Correlation kernels} \label{sect:kernel}

In this section, we review the main properties of the correlation kernel of the Ginibre ensemble, and give first useful consequences thereof. To unify the presentation we use the Lowest Landau Level~\eqref{eq:intro LLL} as our starting point. This is a special case (for a Gaussian weight) of the usual Bergman space used in the context of (planar) orthogonal polynomials  and normal random matrices, e.g.~\cite{AmeHedMak-11}.

\subsection{The lowest Landau level and its kernel} \label{sect:Kprop}

Using the lowest Landau level orbitals from~\eqref{eq:orbitals}, we can rewrite the (Ginibre) correlation kernel from~\eqref{def:KN} as
$$
\K_N(z,w) = \sum_{j=0} ^{N-1} \varphi_j (z) \overline{\varphi_j (w)} .
$$
Hence, $\K_N$ is the integral kernel of the $L^2$-orthogonal projector onto the subspace of the lowest Landau level~\eqref{eq:intro LLL} spanned by eigenfunctions of the angular momentum with eigenvalues less than or equal to $N-1$. 
The natural large-$N$ bulk limit of $\K_N$ is 
\begin{equation} \label{def:K}
 \K_\infty(z,w) =  \frac{\B}{\pi} e^{-\frac{\B}{2}\left( |z|^2 + |w|^2 - 2 z \overline{w}\right)},
\end{equation}
the integral kernel of the $L^2$-orthogonal projector onto the full $\LLL$; see e.g.~\cite{RouYng-19} and references therein. We record some straightforward properties that we will extensively use in the sequel: 

\begin{properties}[\textbf{The Bergman kernel $\K_\infty$}]\label{pro:Kinfty}\mbox{}\\ With $\B=N$, we have the following identities:
\begin{equation} \label{Kmodule}
| \K_\infty(z,w) | = \frac{N}{\pi} e^{-N|z-w|^2/2}.
\end{equation} 
More generally, any (mixed) derivative of the kernel  $\K_\infty$ equals a polynomial in $(z,w)$ times $\K_\infty (z,w)$. In particular, for any multi-index $\alpha \in \N_0^4$, uniformly for all $z,w \in \D$ with $|z-w| \ge 2\delta_N= 2\kappa \sqrt{\frac{\log N}{N}}$, we have
\[
\big| \partial_{\overline{z}}^{\alpha_1} \partial_z^{\alpha_2}   \partial_{\overline{w}}^{\alpha_3}  \partial_w^{\alpha_4} \,\K_\infty(z,w) \big| = \O_\alpha\big(N^{1+|\alpha|-2\kappa^2}\big) . 
\]
where $|\alpha| = \alpha_1  + \alpha_2 + \alpha_3 + \alpha_4$.
\end{properties}

Since $K_N,K_\infty$ are integral kernels of orthogonal projectors we also have the following formulae:

\begin{properties}[\textbf{Reproducing kernels}]\label{pro:reprod}\mbox{}\\
For any $N \in \{1,2,\cdots, \infty \}$ and for any  (analytic) polynomial $f$ of degree (strictly) less than $N$,
\begin{equation*} \label{reproducing0}
\int_\C\K_N (z,w) f(w) e^{-\frac{\B}{2} |w|^2} dw = f(z) e^{-\frac{\B}{2} |z|^2} , \qquad z\in \C.
\end{equation*}
Hence, for all $z,w\in \C$
\begin{equation} \label{reproducing1}
\int_\C \K_N(z,x) \K_N(x,w) \,dx  = \K_N (z,w). 
\end{equation}
Moreover, if $f(z)$ is a polynomial of degree $\le N$ such that $f(x) =0$, we have for any $x\in\C$,
\begin{equation} \label{reproducing4}
\int_\C \frac{f(w)}{w-x}  \K_N(z,w) e^{-\frac{\B}{2} |w|^2} dw  = \frac{f(z)}{z-x} e^{-\frac{\B}{2} |z|^2} 
\end{equation}
where the right-hand side is interpreted as $f'(x)e^{-\frac{\B}{2} |x|^2}$ when $z=x$, with $f' = \partial_z f$.
\end{properties}

We also easily verify that for any $\B>0$ and $N\in\N$, 
\begin{equation} \label{Kasymp2}
0< \K_N(z,z) \le  \K_\infty (z,z)    = \frac{\B}{\pi} \qquad \text{for all } z\in\C
\end{equation}
and 
\begin{equation} \label{reproducing2}
\int_\C \K_N(z,z)  dz = N, \quad \int_\C |\K_N(z,w)|^2 dz dw  = N. 
\end{equation}
The latter integrals are indeed respectively the trace-class and Hilbert-Schmidt norms of the corresponding orthogonal, rank $N$, projectors. Finally, we record that by the Cauchy--Schwarz inequality, we have for any $N\in\N$ and $z,w\in\C$, 
\begin{equation} \label{reproducing3}
|\K_N(z,w) |^2 \le \K_N(z,z) \K_N(w,w) \leq \frac{\B^2}{\pi^2} . 
\end{equation}

\subsection{Large $N$ asymptotics of correlation kernels} \label{sect:Kasymp}

Recall that  
$$\delta_N(\kappa) = \kappa \sqrt{\frac{\log N}{N}} < 1$$
 for $\kappa\ge 0$ and that we always set $\B= N$ sufficiently large. Up to the extra factor $\sqrt{\log N}$, this corresponds to the \emph{microscopic scale} that is, the inter-particle distance inside the plasma.

\begin{lemma}[\textbf{Large $N$ asymptotics of correlation kernels}]\label{lem:ker}\mbox{}\\
For any $\delta\in(0,1]$ (possibly depending on $N>0$), we have for any $z,w\in\C$ with $|w| \le 1-\delta$, 
\begin{equation} \label{Kasymp0}
|\K_{N+n}(z,w)| \le \frac{N}{\pi} e^{-N( |z| - 1+\delta)_+^2/2} . 
\end{equation}
Moreover, for any multi-index $ \alpha \in \N_0^4$, uniformly for all  $|z|,|w|  \le 1-  \delta_N(\kappa)$, we have 
\begin{equation} \label{Kasymp1}
\partial_{\overline{z}}^{\alpha_1} \partial_z^{\alpha_2}   \partial_{\overline{w}}^{\alpha_3}  \partial_w^{\alpha_4} \, \K_{N+n}(z,w) 
= \partial_{\overline{z}}^{\alpha_1} \partial_z^{\alpha_2}   \partial_{\overline{w}}^{\alpha_3}  \partial_w^{\alpha_4} \,\K_\infty(z,w) + \cO_\alpha\big(N^{1+|\alpha|-2\kappa^2}\big) 
\end{equation}
where $|\alpha| = \alpha_1  + \alpha_2 + \alpha_3 + \alpha_4$.
\end{lemma}

To prove our main results, we can choose $\kappa$ arbitrary large so that all remainders in the estimates above can thus be thought of $\O(N^{-\infty})$ as $N\to\infty$ that is, they decay faster than any (negative) power of $N$. For completeness we state and prove nearly optimal bounds, but if one restricts further inside the droplet, e.g. to $|z|,|w| \leq 1 - N^{-1/2 + \epsilon}$, the proof of these estimates is even more straightforward. 

Lemma~\ref{lem:ker} will allow to replace $\K_{N+n}$ (along with any of its derivatives) by the explicit kernel $\K_\infty$ in the sequel. 
Then, any derivative of $\K_\infty$ can be easily controlled using Properties~\ref{pro:Kinfty}. For instance, by \eqref{def:Upsilon},  Lemma~\ref{lem:ker} with $\alpha=0$ implies that for any configuration of quasi-holes $\w \in \DROPn$, cf.~\eqref{eq:config}, 
\[
\Upsilon_N(\w) = \Upsilon_\infty(\w) + \cO(N^{-2\kappa^2})
\]
where $\Upsilon_\infty(\w)$ is defined as in~\eqref{def:Upsilon} but with $\K_{N+n}$ replaced by $K_\infty$.

\begin{proof}
We already know that  $|\K_{N+n}(z,w)| \le N$  for all $z,w\in\C$. Hence, to prove \eqref{Kasymp0}, it suffices to consider the case $|z|\ge 1-\delta \ge |w|$. 
 By the triangle inequality, 
 \begin{align*} \notag
|\K_{N+n}(z,w)|   & \le \frac{N}{\pi} e^{-N( |z|^2 +|w|^2)/2} \sum_{k=0}^{N+n-1} \frac{N^k|zw|^k}{k!}   \\
&\le  \frac{N}{\pi} e^{-N( |z|^2 +|w|^2 - 2|z||w|)/2}  = \frac{N}{\pi} e^{-N( |z|-|w|)^2 /2}  . 
\end{align*}
This immediately implies \eqref{Kasymp0}. 

In order to obtain the asymptotics \eqref{Kasymp1}, we observe that by the residue formula, we have for any $x\in\C$ and any $\epsilon>0$, 
\[ \begin{aligned}
\sum_{j=N+n}^{+\infty} \frac{x^j}{j!}  & =  \frac{1}{2\pi\i}\sum_{j=N+n}^{+\infty}\oint_{|\xi|=1+\epsilon} e^{x \xi} \frac{\d\xi}{\xi^{j+1}}  \\
& =  \frac{1}{2\pi\i}  \oint_{|\xi|=1+\epsilon}  \frac{e^{x \xi}}{\xi-1} \frac{\d\xi}{\xi^{N+n}} . 
\end{aligned}\]
Let us define 
$$\varphi(z) = z-\log(z)-1$$
where $\log(\cdot)$ denotes the principal branch. 
The previous formula implies that for any $x\in\C$ and any $\epsilon>0$, 
\begin{equation} \label{Kasymp5} 
\begin{aligned}
\sum_{j=N+n}^{+\infty} N^{j+1}\frac{x^j}{j!} & =  \frac{N}{2\pi\i}\oint_{|\xi|=1+\epsilon}  \frac{e^{N x \xi}}{\xi-1} \frac{\d\xi}{\xi^{N+n}} \\
& =  \frac{N}{2\pi\i} x^{N+n} \oint_{|\xi|=|x|+\epsilon}  \frac{e^{Nz}}{z-x} \frac{\d z}{z^{N+n}} \\
& =  \frac{N}{2\pi\i} x^{N+n} e^N \oint_{|z|=|x|+\epsilon}  \frac{e^{N\varphi(z)}}{z-x} \frac{\d z}{z^{n}} . 
\end{aligned}
\end{equation}
We use Laplace's method to compute the asymptotics of this integral. We verify that the function $\varphi$ has a unique saddle point at $z=1$ with $\varphi(1)=0$ and $\varphi''(1)=1$ and that (by convexity) we have for any $\theta\in[-\pi, \pi]$, 
\[
-\Re \varphi(e^{\i\theta}) = 1-\cos(\theta) \ge \theta^2/5 .
\]
This implies that for any $|x|<1$, 
\begin{equation} \label{Laplace_asymp}
\left| \oint_{|z|=1}  \frac{e^{N\varphi(z)}}{z-x} \frac{\d z}{z^{n}}  \right|  \le \int_{-\pi}^{\pi} \frac{e^{-N\theta^2/5}}{|1- xe^{-\i \theta}|} d\theta  \le \sqrt{\frac{5\pi}{N}} \frac{1}{1-|x|} . 
\end{equation}
By combining \eqref{Kasymp5} and \eqref{Laplace_asymp}, we have shown that for any $|x|<1$ 
\begin{equation} \label{exp_asymp}
\left| \sum_{j=N+n}^{+\infty} N^{j+1}\frac{x^j}{j!}  \right| \le  \sqrt{\frac{5N}{4\pi}} \frac{|x|^{N+n} e^N}{1-|x|}
\end{equation}
By definition, observe that 
\begin{equation} \label{Kdiff}
\K_\infty(z,w) - \K_{N+n}(z,w) = \frac{N}{\pi}e^{-N( |z|^2 +|w|^2)/2} \sum_{j=N+n}^{+\infty} N^{j}\frac{x^j}{j!}
\end{equation}
 with $x=z\overline{w}$. 
Thus, this implies that  for any $z,w\in\C$ with $|z\overline{w}|<1$, 
\begin{align} \label{exp_asymp_2}
\Big| \K_\infty(z,w) - \K_{N+n}(z,w) \Big|
\le \frac{1}{\pi} \sqrt{\frac{5N}{4\pi}}\frac{1}{1-|z\overline{w}|}  e^{-N\big( \varphi(|z|^2)+ \varphi(|w|^2) \big)/2} ,
\end{align}
where we used that 
$$\varphi(|z|^2) = |z|^2-1-2\log|z|.$$
We have $\varphi(|z|^2) \ge 0$ for all $z\in \C$ and 
\begin{equation} \label{varphi_bd}
N\varphi(|w|^2) \ge 2\kappa^2 \log N +\O\bigg(\frac{\kappa^3(\log N)^{3/2}}{\sqrt{N}}\bigg) \qquad\text{if}\qquad |w| \le 1- \delta_N(\kappa). 
\end{equation}
We conclude that, uniformly for all $|z| , |w| \le 1- \delta_N(\kappa)$, 
\begin{equation} \label{kernelbound}
\Big| \K_\infty(z,w) - \K_{N+n}(z,w) \Big| =\O\bigg(\frac{N^{1-2\kappa^2}}{\kappa\sqrt{\log N}}\bigg) .
\end{equation}
This completes the proof of \eqref{Kasymp1} in case $\alpha=0$. 
The estimates on the derivatives of the kernel are obtained in the exact same way by differentiating formulae \eqref{Kdiff} and using the bound \eqref{exp_asymp}. In particular, we simply pay a factor of $N$ in the error term for each derivative. 
\end{proof}

\begin{remark}\label{rk:exp_asymp}
The bound \eqref{exp_asymp} shows that for all $x\in[0,1)$, 
\[
\sum_{j=N+n}^{+\infty} N^{j+1}\frac{x^j}{j!}  e^{-Nx} 
\le \sqrt{\frac{5N}{4\pi}} \frac{ e^{-N \varphi(x)}}{1-x} , 
\]
where $\varphi(x) = x-\log(x)-1$. There is a probabilistic interpretation for this bound. We have for any $x>0$, 
\[
\sum_{j<N+n} N^{j}\frac{x^j}{j!}  e^{-Nx}  = \P[ X_1+ \cdots + X_{N+n-1} \le xN]
\]
where $(X_k)_{k=0}^{+\infty}$ are i.i.d. exponential random variables with mean $1$. 
Then by Cram\'er's theorem, it holds that for any $x<1$, 
\[
\limsup_{N\to+\infty}\frac{1}{N}\log\bigg(\sum_{j<N+n} N^{j}\frac{x^j}{j!}  e^{-Nx}\bigg) \le  - \varphi(x)
\]
where $ \varphi(x)= \sup_{t > -1}\big\{\log(1+t)-tx \big\} = \varphi(x)$ corresponds to the Legendre transform of the cumulant generating function of an exponential random variable. 
\end{remark}

\subsection{First bounds for the remainder terms in the partition function}

We conclude this section by deducing asymptotics for the partition function $\cqh(\w)$ for $\w \in \DROPn^{\varnothing}$, cf.~\eqref{eq:nomerg}. We start from Theorem \ref{thm:exp} and recall its' rewriting as \eqref{cqh}. Our task is to prove that $\Upsilon_N(\w) \simeq 1$ with sufficient precision, in particular that the derivatives of $\Upsilon_N$ can be mostly neglected when computing the emerging fields as in \eqref{eq:vector pot}-\eqref{eq:scalar pot}.

\begin{proposition}[\textbf{Bounds on remainders, no merging}] \label{prop:main}\mbox{}\\
Recall the notation \eqref{def:Upsilon} with $\B=N$. 
For any multi-index $\alpha,\beta \in \N_0^{n}$ and uniformly for all $\w\in \DROPn^{\varnothing}$, we have 
\[
\partial_{\w}^\alpha  \partial_{\overline{\w}}^\beta \Upsilon_N(\w) = \delta_{\alpha=0}\delta_{\beta=0}+\cO\big( N^{|\alpha| + |\beta| -2\kappa^2}\big)  . 
\]
\end{proposition}

\begin{proof}
By definition
$$ \Upsilon_N (\w) =  \left( \frac{\pi}{N} \right)^n\sum_{\sigma \in \S_n} \mathrm{sgn}(\sigma) \prod_{i=1}^n \K_{N+n}\big(w_i,w_{\sigma(i)}\big)$$
where the sum is over all permutations of $n$ elements. Hence, by Lemma~\ref{lem:ker}, we have that, for any multi-index $\alpha,\beta \in \N_0^{n}$, uniformly for $\w \in \DROPn$, 
\begin{equation} \label{Lapexp}
\partial_{\w}^\alpha  \partial_{\overline{\w}}^\beta\Upsilon_N(\w) =
  \left( \frac{\pi}{N} \right)^n\sum_{\sigma \in \S_n} \mathrm{sgn}(\sigma) \,
  \partial_{\w}^\alpha \partial_{\overline{\w}}^\beta \Big( \prod_{i=1}^n \K_{\infty}\big(w_i,w_{\sigma(i)}\big) \Big) + \cO\big(N^{|\beta|+|\alpha|-2\kappa^2}\big) . 
\end{equation}
Then, using Properties~\ref{pro:Kinfty}, for $\w \in \DROPn^{\varnothing}$, the main contribution to this sum comes from the identity permutation, that is, 
\[
\partial_{\w}^\alpha  \partial_{\overline{\w}}^\beta\Upsilon_N(\w) =
  \partial_{\w}^\alpha \partial_{\overline{\w}}^\beta \Big( \prod_{i=1}^n \tfrac{\pi}{N}  \K_{\infty}\big(w_i,w_i\big) \Big) + \cO\big(N^{|\beta|+|\alpha|-2\kappa^2}\big) . 
\]
Since, $\K_\infty (z,z) = N/\pi$ is constant the leading term equals $1$ if $\alpha=\beta =0$ and vanishes otherwise. This completes the proof, with the required uniformity. 
\end{proof}

This argument also allows to obtain the asymptotics of $\Upsilon_N$ (along with its derivatives) when quasi-holes are merging. 
In case $|w_i-w_j| \le 2\delta_N$, we cannot a priori drop the terms $\K_\infty(w_i,w_j)$ from the sum \eqref{Lapexp}. In general, this gives rise to a challenging combinatorial problem to keep track of all non-trivial terms.  
However,  it is easy to keep track of the correction terms if we restrict to the case where \emph{only} two quasi-holes are allowed to come \emph{microscopically close}. This is referred as \emph{simple merging} in the sequel.

For (fixed) parameters $\kappa, \gamma >0$, let us define the subsets in the complement of $\DROPn^{\varnothing}$ (cf~\eqref{eq:nomerg}), 
for any (unordered) pair  $\{i,j\}\in\{1,\dots n\}$, 
\begin{equation*}\label{eq:onemerge}
\DROPn^{ij} := \big\{ \by \in \DROPn : |y_i-y_j| \le 2 \delta_N(\kappa) \text{ and } |y_k-y_\ell| \ge 2 \delta_N(\kappa) \text{ for } \{k,\ell\} \neq \{i,j\} \big\},
\end{equation*}
with a single merging pair. Without loss of generality we pick the pair $(1,2)$ of particles as the colliding ones, i.e. we work on the set $\DROPn^{12}$. Compared to Proposition~\ref{prop:main}, our next result includes the \emph{correction} which arises from the quasi-holes merging in $\DROPn^{12}$:

\begin{proposition}[\textbf{Refined remainders on single merging}] \label{prop:upgrade}\mbox{}\\
For any multi-index $\alpha,\beta \in \N_0^{n}$ and uniformly for $\w\in \DROPn^{12}$, we have 
\begin{equation} \label{Upscorr}
\partial_{\w}^\alpha  \partial_{\overline{\w}}^\beta\Upsilon_N(\w)
= \delta_{\alpha=0}\delta_{\beta=0} - \partial_{\w}^\alpha \partial_{\overline{\w}}^\beta  \big(e^{-N|w_1-w_2|^2} \big) + \cO\big(N^{|\beta|+|\alpha|-2\kappa^2}\big) . 
\end{equation}
\end{proposition}

\begin{proof}
Starting from the expansion \eqref{Lapexp} and using the Properties~\ref{pro:Kinfty} of the kernel $\K_\infty$,  for $\w\in \mathscr{D}_{n,1}^{\varnothing}$, the main contribution to this sum comes from the permutation
$$\sigma\in \S_n : \sigma(1)= 2, \sigma(2)=1 \text{ and } \sigma(i) =i \text{ for all } i \notin\{1,2\}.$$
This implies that, uniformly for  $\w\in \mathscr{D}_{n,1}^{\varnothing}$, 
\[\begin{aligned}
\partial_{\w}^\alpha  \partial_{\overline{\w}}^\beta\Upsilon_N(\w)
&  =  \partial_{\w}^\alpha \partial_{\overline{\w}}^\beta \Big( \prod_{i=1}^n  \tfrac{\pi}{N}  \K_{\infty} (w_i,w_i) \Big)   \\
& \quad -  \partial_{\w}^\alpha \partial_{\overline{\w}}^\beta \Big( |  \tfrac{\pi}{N} \K_\infty(w_1,w_2)|^2 \prod_{i\notin \{1,2\}}   \tfrac{\pi}{N} \K_{\infty} (w_i,w_i) \Big) + \cO\big(N^{|\beta|+|\alpha|-2\kappa^2}\big)
\end{aligned}\]
where the first term comes from the identity permutation. Using \eqref{Kmodule} and that on the diagonal $\K_\infty (z,z) = N/\pi$ is constant $\K_\infty$ is constant and equals $N/\pi$, we immediately obtain the claimed asymptotics.
\end{proof}

The asymptotics \eqref{Upscorr} hold on the set 
\[
\big\{ \by \in \DROPn : |y_k-y_\ell| \ge 2 \delta_N(\kappa) \text{ for } (k,\ell) \neq (1,2) \big\} ,
\]
however, the correction term is only relevant if  $|y_j-y_1| \le \delta_N$. Otherwise, we recover the estimates from Proposition~\ref{prop:main}.

\section{Estimates of the emerging potentials}\label{sec:estimates}

We now connect the exact expressions of the emerging potentials obtained in Section~\ref{sec:calculus} with Theorem~\ref{thm:transmutation}, thereby putting the heuristics of Section~\ref{sect:proof out} on rigorous ground. In contrast with the calculations of Section~\ref{sec:calculus}, the sequel heavily relies on the determinantal structure of the Ginibre ensemble.

We shall use three types of bounds, presented in separate subsections. We consider first  the case where the quasi-holes are all well-separared, then the case of simple quasi-holes encounters/merging. More complicated mergings will be controlled by uniform (non-optimal) bounds on the emerging potentials obtained in Section~\ref{sect:global} below.

\subsection{Asymptotics away from merging}\label{sect:no merging} 

The core of the statistics transmutation phenomenon is given by the following

\begin{proposition}[\textbf{Vector and scalar potentials away from merging}]\label{prop:nomerg}\mbox{}\\
Recall the definition ~\eqref{eq:nomerg} of the no-merging set $\DROPn^{\varnothing}$ and assume $\kappa \geq 2$. Let $\cA_j,\cV_j$ be the emerging vector and scalar potentials obtained in Proposition \ref{prop:exp}. It holds for every  $j=1,\dots,n$, uniformly for all $\by\in \DROPn^{\varnothing}$ 
\begin{equation}\label{eq:Anomerg}
\mathcal{A}_j (\by) = N w_j^{\perp} - \sum_{\ell=2}^n\frac{(y_1-y_\ell)^\perp}{|y_1-y_\ell |^2} +  \O\big( N^{1-2\kappa^2} \big)
\end{equation}
and 
\begin{equation}\label{eq:Vnomerg}
\mathcal{V}_j (\by)   = 2N +  \O\big( N^{2-2\kappa^2}\big)  .
\end{equation}
\end{proposition}

\begin{proof}
This is a direct consequence of Theorem~\ref{thm:exp} and the estimates from Proposition~\ref{prop:main}.
For simplicity, we use complex coordinates. 
Combining the second expressions in ~\eqref{eq:vector pot bis}-\eqref{eq:scalar pot bis}  with formula~\eqref{cqh}, we find by  direct computations that, when $w_i\neq w_j$ for all $i\neq j$, 
\begin{equation} \label{eq:pot remainders}
\begin{aligned}
\mathcal{A}_j (\w) &= N w_j^{\perp} - \sum_{\ell=2}^n\frac{(w_1-w_\ell)^\perp}{|w_1-w_\ell |^2} +  \Im \Big(  {\footnotesize	 \begin{pmatrix} 1 \\ \i \end{pmatrix}} \partial_{w_j}  \log \Upsilon_N(\w) \Big) \\
\mathcal{V}_j (\w) &= 2N  + 2 \partial_{\overline{w_j}}\partial_{w_j} \log \Upsilon_N(\w).
\end{aligned}
\end{equation}
where we used~\eqref{eq:Lap}-\eqref{eq:grad 2} ($\Upsilon_N(\w)$ is real-valued). 
We also used that the Vandermonde determinant $\triangle(\w)$ is a poly-analytic function
which does not vanish on $\{\w\in \C^n : w_1\neq \cdots \neq w_j \}$, so that  $\Delta_{w_j}\log \left| \triangle(\w)\right| =0$ on this set. 
 
 Using Proposition~\ref{prop:main}, we conclude that for any $j\in\{1,\dots, n\}$,
 \[
\partial_{w_j}  \log \Upsilon_N(\w)=  \Upsilon_N^{-1}(\w)  \partial_{w_j}   \Upsilon_N(\w)  =  \O\big( N^{1-2\kappa^2} \big)
\]
and 
\[
\partial_{\overline{w_j}}\partial_{w_j} \log \Upsilon_N(\w) =
\Upsilon_N^{-1}(\w)  \Delta_{w_j}   \Upsilon_N(\w) - |\partial_{w_j}  \log \Upsilon_N(\w)|^2  =  \O\big( N^{2-2\kappa^2} \big)
\]
as claimed. 
\end{proof}

\subsection{Simple merging}

We now consider the asymptotics of the emerging potentials in the (simple) merging regions \eqref{def:merg set}. Correction terms can be identified using the functions $\mathrm{a} : \R^2 \to \R^2$ and $\mathrm{v} : \R^2 \to \R$,  
\begin{equation} \label{corrpot}
 \begin{aligned}
\mathrm{a}(y)  & = \frac{y^\perp}{e^{|y|^2}-1} = \nabla^\perp\log (1-e^{-|y|^2}) \\
\mathrm{v}(y)& =  2\frac{1-(1-|y|^2) e^{|y|^2}}{(e^{|y|^2}-1)^2} =   - \frac12 \Delta  \log (1-e^{-|y|^2}) . 
\end{aligned}
\end{equation}
Note that $\mathrm{v}(y) \frac{\d y}{\pi}$ is a probability measure on $\R^2$ and that $\mathrm{a}\in L^p(\R^2)$ for $1\leq p <2$. These functions have the following asymptotics as $y\to0$, 
\[
\mathrm{a}(y) \simeq \frac{y^\perp}{|y|^2}
\qquad\text{and}\qquad 
\mathrm{v}(y) \simeq 1. 
\]
It will be convenient to assume that the colliding particles do not get absurdly close. To this end we introduce a parameter $\gamma >0$,  and mainly work on   
\begin{equation} \label{def:merg set}
\DROPn^{ij,\gamma} : = \left\{ \by \in \DROPn^{ij},  
 |y_i - y_j|^2  \ge N^{-1-\gamma}\right\}.  
\end{equation} 
We ultimately choose large fixed parameters $\gamma, \kappa$ under some appropriate condition. The remainder of the simple merging set $\DROPn^{ij} \setminus \DROPn^{ij,\gamma}$ will be dealt with using the estimates from Section~\ref{sect:global} below.

\begin{proposition}[\textbf{Vector and scalar potentials with simple merging}]\label{prop:merg}\mbox{}\\
For any $\kappa, \gamma \ge 1$ satisfying $\gamma < 2 \kappa^2$, uniformly for all $\w \in \DROPn^{12,\gamma}$, we have
\begin{equation} \label{corrasymp}
\begin{aligned}
\mathcal{A}_j (\by)  & = N y_j^{\perp}  - \sum_{\ell=2}^n\frac{(y_j-y_\ell)^\perp}{|y_j-y_\ell |^2}  +  
\1_{j\in\{1,2\}} \sqrt{N} \mathrm{a}\big(\sqrt{N}(y_1-y_2) \big) + \O(N^{1+2\gamma-2\kappa^2})  , \\
\mathcal{V}_1(\by)  & = N \big(2 - \1_{j\in\{1,2\}}  \mathrm{v}\big(\sqrt{N}(y_1-y_2)\big) \big) + \O(N^{1+3\gamma -2\kappa^2}) .
\end{aligned}
\end{equation}
\end{proposition}

Both correction terms  are functions of the microscopic variable $y=\sqrt{N}(y_1-y_2)$ for the merging pair $\{y_1,y_2\}$ and decay like $e^{-N|y_1-y_2|^2}$ for large $|y_1-y_2|$.

\begin{proof}
We proceed as in the proof of Proposition~\ref{prop:nomerg}, starting from~\eqref{eq:pot remainders}. By Proposition~\ref{prop:upgrade}, we verify that  for $\w \in \DROPn^{12,\gamma}$ and any $j\in\{1,\dots,n\}$
\[\begin{aligned}
\partial_{w_j}  \log \Upsilon_N(\w)  & = \big( - \partial_{w_j}  \big(e^{-N|w_1-w_2|^2} \big) +\O(N^{1-2\kappa^2}) \big) \Upsilon_N^{-1}(\w) \\
& = \partial_{w_j}\big(  \log (1-e^{-N |w_1-w_2|^2})  \big)  \frac{1-  e^{-N|w_i-w_j|^2}}{\Upsilon_N(\w)} +\O\big(\Upsilon_N^{-1}(\w) N^{1-2\kappa^2} \big). 
\end{aligned}\]
But, since $\gamma<2\kappa^2$ in~\eqref{def:merg set}, we deduce from Proposition~\ref{prop:upgrade} again that 
\begin{equation} \label{LB1}
\Upsilon_N(\w)  = 1-   e^{-N|w_i-w_j|^2} +\O(N^{-2\kappa^2}) \ge  c N^{-\gamma}
\end{equation}
on $\DROPn^{ij,\gamma}$, for a constant $c$ depending only on $(\gamma,\kappa,n)$. This implies that 
\begin{equation} \label{LB2}
\frac{1-  e^{-N|w_i-w_j|^2} }{\Upsilon_N(\w)}
= 1 +\O(N^{\gamma-2\kappa^2}) .  
\end{equation}
Inserting in the above, using that $|\mathrm{a}(y)| \le C/|y|$ for $y\in\R^2$ and that $\gamma \geq 1$, we obtain the first line in~\eqref{corrasymp}. We also used the expression \eqref{corrpot} of $\mathrm{a}$ to identify the correction term.

Next, we turn to the scalar potential. Recall that for  $j\in\{1,\dots,n\}$, 
\[
 \partial_{\overline{w_j}}\partial_{w_j} \log \Upsilon_N(\w)
 = \Upsilon_N^{-1}(\w)  \partial_{\overline{w_j}}\partial_{w_j} \Upsilon_N(\w)
- |  \Upsilon_N^{-1}(\w) \partial_{w_j}  \Upsilon_N(\w)|^2 .
\]
The second term is the square of that we dealt with above. As for the second, using Proposition~\ref{prop:upgrade} and~\eqref{LB1}-\eqref{LB2} again, an analogous computation shows that for $\w \in \DROPn^{12,\gamma}$, 
\[\begin{aligned}
 \partial_{\overline{w_j}}\partial_{w_j} \log \Upsilon_N(\w)
=  \partial_{\overline{w_j}}\partial_{w_j} \big(  \log (1-e^{-N |w_1-w_2|^2})  \big) +  \O(N^{1+3\gamma-2\kappa^2})
 \end{aligned}\]
This completes the proof of~\eqref{A1}. 
\end{proof}

\subsection{Global bounds}\label{sect:global}

In order to discard the contribution of the small part of configuration space where more than two quasi-holes are close, it is sufficient to have at our disposal the following uniform bounds: 

\begin{proposition}[\textbf{Global bounds on emerging potentials}] \label{prop:global}\mbox{}\\
Let $\cA_j,\cV_j$ be the emerging vector and scalar potentials obtained in Proposition \ref{prop:exp}.
There exists a constant $C$ depending only on $n$ such that for any  $j=1,\dots,n$
\begin{equation}\label{eq:Aglob}
\sup_{\by \in \R^{2n}} | \mathcal{A}_j (\by)| \leq C N ,
\end{equation}
and 
\begin{equation}\label{eq:Vglob}
\sup_{\by\in \C^n} | \mathcal{V}_j (\by) | \leq C N^{3/2} .
\end{equation}
\end{proposition}

\begin{proof}
We use the integral expressions given in Proposition \ref{prop:int}. Let us begin by some apriori estimates for $\Upsilon$. We denote
\[
\M(\w) = \big[\K_{N+n}(w_j,w_i) \big]_{n \times n} , \qquad \w\in \C^n .
\]
According to \eqref{def:Upsilon},
$$
 \Upsilon (z,\w) = \left(\frac{\pi}{N}\right)^{n+1} \det \M (z,\w) = \left(\frac{\pi}{N}\right)^{n+1} \det \begin{pmatrix}
\K_{N+n} (z,z)  & \nu^* (z|\w) \\
\nu (z|\w) & \M (\w)
\end{pmatrix} 
$$
where the vector $\nu (z|\w)$ is 
\begin{equation}\label{eq:SchurVec}
\nu(z|\w) : =    \begin{pmatrix} \K_{N+n} (z,w_1) \\ \vdots \\ \K_{N+n}(z,w_n) \end{pmatrix}.
\end{equation}
The Schur complement formula then gives, 
\begin{equation} \label{def:theta}
\Upsilon(z,\w) = \frac{\pi}{N}  \Upsilon(\w)\big( \K_{N+n} (z,z)- \Theta(z|\w) \big)
\quad\text{where}\quad 
 \Theta(z|\w) := \nu^*(z|\w) \M(\w)^{-1} \nu(z|\w)  . 
\end{equation}
Since $\M(\w)$ is a positive definite matrix ($\K$ is a correlation kernel)  for $\w \in \{\C^n :w_1\neq\cdots \neq w_n\}$, we have the bound 
\[
0 \le  \Theta(z|\w)  \le  \K_{N+n} (z,z) \le N/\pi . 
\]
Moreover $z\in\C \mapsto \Theta(z|\w) $ is a density function with (fixed) mass $n$ since
\[
\int_\C  \Theta(z|\w)  \d z  =  \tr\bigg[ \M(\w)^{-1} \int_\C  \nu(z|\w)\nu^*(z|\w) \d z  \bigg] 
=  \tr\big[ \M(\w)^{-1} \M(\w) \big] = n  
\]
where we used the reproducing property (Property~\ref{pro:reprod}). 
In particular, these two estimates for the density $\Theta(z|\w)$ imply that for any $\w\in\C^n$
\begin{equation} \label{bdTheta}
0 \le \int_\C \frac{\Theta(z|\w)}{|z-w_j|} \d z \le n \sqrt{N} . 
\end{equation}
We obtain this bound by splitting the integral over the regions $\big\{|z-w_j|\le 1/\sqrt{N}\big\}$ and 
 $\big\{|z-w_j|\ge 1/\sqrt{N}\big\}$. 

\medskip 
We are now ready to turn to the proofs of the claimed uniform bounds for the emerging potentials.
\medskip

\noindent \textbf{Proof of \eqref{eq:Aglob}.} 
It follows by combining \eqref{eq:vect int} and \eqref{def:theta}  that for any $\w\in\C^n$. 
\begin{equation} \label{bdA}
 | \mathcal{A}_j (\w) | \leq \sqrt{2} \int_\C \frac{ \K_{N+n}(z,z) }{\vert z-w_j \vert} \d z 
\le C N . 
\end{equation}
This bound easily follows by separating between regions with $\vert  z-w_j \vert \leq 1$ and $\vert  z-w_j \vert \geq 1$, using \eqref{Kasymp2} in the former and \eqref{reproducing2} in the latter.

\medskip 

\noindent \textbf{Proof of \eqref{eq:Vglob}.} We estimate separately the two terms in \eqref{eq:scal int}, namely 
\begin{equation} \label{I0bd} 
\begin{aligned}
I_1&= \frac{2N}{\Upsilon(\w)} \int_\C  \frac{\Upsilon(\w,z)}{|w_j-z|^2} \frac{\d z}{\pi} ,  \\
I_2 &=  \frac{2N^2}{\Upsilon(\w)^2} \iint_{\C^2} \frac{\Upsilon(\zeta,z,\w)\Upsilon(\w)-\Upsilon(z,\w) \Upsilon(\zeta,\w) }{ (w_j-z)(\overline{w_j-\zeta})} \frac{\d z}{\pi} \frac{\d \zeta}{\pi} .
\end{aligned}
\end{equation}
Let us begin with $I_2$. Using~\eqref{def:theta} again, one can rewrite 
\[
\Upsilon(\zeta,z,\w)\Upsilon(\w)-\Upsilon(z,\w) \Upsilon(\zeta,\w) 
= \frac{\pi}{N}   \Upsilon(\w) \Upsilon(z,\w) \big( \Theta(\zeta|\w) - \Theta(\zeta| z,\w) \big) .
\]
and then bound as $\Theta\ge 0$
\[
\big| \Upsilon(\zeta,z,\w)\Upsilon(\w)-\Upsilon(z,\w) \Upsilon(\zeta,\w)  \big|
\le  \left(\frac{\pi \Upsilon(\w)}{N} \right)^2 \big( \K_{N+n}(z,z)\Theta(\zeta|\w) + \K_{N+n}(z,z)\Theta(\zeta|z,\w) \big) , 
\]
This implies that 
\[ \begin{aligned}
I_2  
& \le 2 \int_{\C}  \frac{\K_{N+n}(z,z)}{ |w_j-z|}  \bigg( \int_\C \frac{\Theta(\zeta|\w)}{ |w_j-\zeta|} \d \zeta + \int_\C  \frac{\Theta(\zeta|z,\w)}{ |w_j-\zeta|} \d\zeta \bigg)  \d z
\end{aligned}\]
Using the uniform estimate \eqref{bdTheta}, we conclude that 
\begin{equation} \label{I2bd} 
I_2  \le C \sqrt{N} \int_{\C}  \frac{\K_{N+n}(z,z)}{ |w_j-z|} \d z  \le C N^{3/2}
\end{equation}
by \eqref{bdA}. 

\medskip

We now turn to estimate the term $I_1$ from~\eqref{I0bd}. We generalize the definition \eqref{def:theta} to let  
\begin{equation} \label{theta_polarization}
 \Theta(\zeta,z|\w) : =  \nu^*(z|\w) \M(\w)^{-1} \nu(\zeta|\w) 
\end{equation}
where $\nu$ is as in \eqref{eq:SchurVec}. We claim that 
\begin{equation} \label{polarization}
I_1 = 2 \int_\C \frac{ \K_{N+n} (z,z)-\Theta(z|\w)}{|z-w_j|^2} d z=  2  \iint_{\C^2}  \frac{ \K_{N+n}(z,\zeta)  \big( \K_{N+n} (\zeta,z)- \Theta(\zeta,z|\w)\big)}{(\zeta-w_j)(\overline{z-w_j)}} d\zeta dz
\end{equation}
The first equality is by definition of $\Theta(z|\w)$. For the second, since $ \M(\w)^{-1} \nu(w_j|\w) = \mathrm{e}_j $ where $(\mathrm{e}_1, \dots , \mathrm{e}_n)$ denotes the canonical basis,  we verify that 
$ \Theta(w_j,z|\w) = \K_{N+n}(w_j,z)$ for all $z\in\C$. Moreover, we can write 
$$ \Theta(\zeta,z|\w) = \mathrm{f}(\zeta)  e^{-\frac{\B}{2}|\zeta|^2}$$
where $\mathrm{f}(\zeta)$ is an (analytic) polynomial of degree $<N+n$. Hence, using the reproducing property~\eqref{reproducing4},  this implies that for all $z\in\C$
\[
\int_\C  \frac{ \K_{N+n}(z,\zeta)  \big( \K_{N+n} (\zeta,z)- \Theta(\zeta,z|\w)\big)}{\zeta-w_j} d\zeta =  \frac{\K_{N+n} (z,z)-\Theta(z|\w)}{z-w_j} 
\]
where we used that by definition, $\Theta(z,z|\w) = \Theta(z|w)$. 
This proves  \eqref{polarization}. 

We can now proceed to estimate $I_1$ as above. We first split 
\begin{equation} \label{I1bd} 
\begin{aligned}
I_1 & \le 2  \iint_{\C^2}  \frac{|\K_{N+n}(z,\zeta)|^2}{\vert z-w_j\vert \vert\zeta-w_j \vert} d\zeta d z 
+ 2  \iint_{\C^2}  \frac{|\K_{N+n}(z,\zeta)| | \Theta(\zeta,z|\w)|}{\vert z-w_j\vert \vert\zeta-w_j \vert} d\zeta d z  \\
&\quad = 2 I_3 + 2 I_4 . 
\end{aligned}
\end{equation}
Since $\M^{-1}(\w)$ is a positive definite matrix, according to \eqref{theta_polarization}, 
\[
| \Theta(\zeta,z|\w)|^2 \le \Theta(z|\w)\Theta(\zeta|\w) . 
\]
Then, by the Cauchy--Schwarz inequality, the second term on the right-hand side of \eqref{I1bd} satisfies
\begin{equation} \label{I4bd} 
I_4
\le  \sqrt{I_3  \iint_{\C^2}  \frac{\Theta(z|\w)\Theta(\zeta|\w)}{\vert z-w_j\vert \vert\zeta-w_j \vert} d\zeta d z }
\le  n \sqrt{N I_3} , 
\end{equation}
using the uniform bound \eqref{bdTheta}. 

Hence, it just remains to bound $I_3$ in \eqref{I1bd}. 
For a small $\epsilon>0$, we define the set 
$$
S_\epsilon = \big\{(z,\zeta) \in \C^2 : |z-w_j| , |\zeta-w_j| \le \epsilon \big\}.
$$
On the complement of $S_\epsilon$, either $ |z-w_j| \ge \epsilon$ or $|\zeta-w_j| \ge \epsilon$, which by symmetry allows to bound
\[
I_3 \le \iint_{S_\epsilon} \frac{|\K_{N+n}(z,\zeta)|^2}{\vert z-w_j\vert \vert\zeta-w_j \vert} d\zeta d z 
+  2 \epsilon^{-1} \iint_{\C^2}  \frac{|\K_{N+n}(z,\zeta)|^2}{\vert z-w_j \vert}  d z d\zeta
\]
Using the bound \eqref{reproducing3}  and the reproducing property \eqref{reproducing1}, this implies that 
\[ \begin{aligned}
I_3 & \le \iint_{S_\epsilon} \frac{N^2}{\vert z-w_j\vert \vert\zeta-w_j \vert} d\zeta d z 
+  2 \epsilon^{-1} \int_{\C}  \frac{\K_{N+n}(z,z) }{\vert z-w_j \vert}  d z  \\
& \le N^2 \epsilon^2+ C \epsilon^{-1} N
\end{aligned}\]
where we used again the estimate \eqref{bdA} to control the second term.

The previous bound holds for any $\epsilon>0$, we can therefore optimize and choose $\epsilon = N^{-1/3}$ which yields
\begin{equation*}
I_3 \le C N^{4/3} . 
\end{equation*}
Hence, by combining this estimate with \eqref{I4bd} and  \eqref{I1bd}, we have shown that 
\begin{equation} \label{I3bd}
I_1 \le C N^{4/3} . 
\end{equation}
According to \eqref{I0bd} and the previous estimates \eqref{I2bd} and \eqref{I3bd}, we conclude that for any $\w\in \C^n$, 
 \[
0\le  \mathcal{V}_j(\w) \le I_2+ I_1 \le  C N^{3/2} . 
\]
This completes the proof. 
\end{proof}

The following bounds improve on Proposition~\ref{prop:global} and they are sharp. 

\begin{proposition}[\textbf{Bounds on the emerging vector potential inside the droplet}] \label{prop:droplet}\mbox{}\\
Let $\cA_j$ be the emerging vector potentials obtained in Proposition \ref{prop:exp}.
For any small $\epsilon>0$, there exists a constant $C_\epsilon$ such that for any  $j=1,\dots,n$
\begin{equation}\label{eq:Aglob2}
\sup_{\by \in \R^{2n} : |y_j| \le 1-\epsilon} |\mathcal{A}_j (\w) - N y_j^\perp | \leq C_\epsilon \sqrt{N} .
\end{equation}
\end{proposition}

\begin{proof}
We start by observing that for any  $w \in \D$, 
\begin{equation} \label{Cauchy}
\overline{w} =  \int_\D \frac{\d z}{\pi(w-z)} 
\end{equation}
Formula \eqref{Cauchy} follows e.g.~from the equilibrium condition for the Ginibre ensemble. By Newton's theorem, on the support of the equilibrium measure (that is, the circular law), there exists a constant $c$ so that
\[
|w|^2 + \int_\D \log|w-z|^{-1} \frac{\d z}{\pi} = c ,
\]
Upon applying $\partial_w$ to this equation, we obtain \eqref{Cauchy}.

Using the integral formula \eqref{eq:vect int} and \eqref{Cauchy}, it follows that for $w_j \in \D$, 
\[ \begin{aligned}
\mathcal{A}_j (\w) - N y_j^\perp & =  N  \Im\bigg(  {\footnotesize	 \begin{pmatrix} 1 \\ \i \end{pmatrix}}  \bigg(\int_\C  \frac{\Upsilon(\w,z)}{\Upsilon(\w)} \frac{\d^2z}{\pi(w_j-z)} - \overline{w_j}  \bigg) \bigg) \\
& = N \Im\bigg(  {\footnotesize	 \begin{pmatrix} 1 \\ \i \end{pmatrix}} \int_\C   \bigg(  \frac{\Upsilon(\w,z)}{\Upsilon(\w)} -\1_\D \bigg) \frac{\d^2z}{\pi(w_j-z)} \bigg) .
\end{aligned}\]
Using again formula \eqref{def:theta} and the bound \eqref{bdTheta}, there is a (numerical) constant $C$ so that
\begin{equation} \label{A2bd}
|\mathcal{A}_j (\w) - N y_j^\perp |  \le 
\sqrt{2}  \int_\C \Big| \K_{N+n}(z,z) - N \frac{\1_\D(z)}{\pi} \Big| \frac{\d^2z}{|w_j-z|}   + C \sqrt{N}
\end{equation}
Using the precise estimate \eqref{exp_asymp_2} and the fact that $\K_\infty(z,z) = N/\pi$, it holds for any $z\in \D$
\begin{equation} \label{A3bd}
\Big| \K_{N+n}(z,z) - N \frac{\1_\D(z)}{\pi} \Big|  \le C \sqrt{N}\frac{e^{-N(1-|z|^2)^2}}{1-|z|^2} 
\end{equation}
where we used that by convexity,
\[
\varphi(x) \ge (1-x)^2 \qquad \text{for all } x\in [0,1] .
\]
In particular, by integrating this estimate, we see that with $\epsilon_N =1/\sqrt{N}$, there is another numerical constant $C$ such that 
\[
\Big| \int_{\{ |z| \le 1- \epsilon_N\}} \hspace{-.3cm}  \K_{N+n}(z,z) \d z  - N(1-\epsilon_N)^2 \Big|  \le C \sqrt{N} 
\]
By \eqref{reproducing2}, this estimate also implies the tail bound
\begin{equation} \label{integralbound}
0 \le \int_{\{ |z| \ge 1- \epsilon_N\}} \hspace{-.3cm}  \K_{N+n}(z,z) \d z \le C \sqrt{N} . 
\end{equation}
For a fixed $\epsilon>0$, combining the two bounds \eqref{A2bd} and \eqref{A3bd}, we obtain for $|w_j| \le 1- \epsilon$, 
\begin{equation} \label{A4bd}
|\mathcal{A}_j (\w) - N y_j^\perp |  \le 
\sqrt{2}  \int_{\{ |z| \ge 1- \epsilon_N\}} \Big| \K_{N+n}(z,z) - N \frac{\1_\D(z)}{\pi} \Big| \frac{\d^2z}{|w_j-z|}   + C_\epsilon \sqrt{N} .
\end{equation}
We also verify that the estimate \eqref{A4bd} still holds if $w_j$ approaches the boundary of the droplet. Namely, for $|w_j| \le 1- \epsilon_N$, the second term is controlled by $C \sqrt{N} \log \epsilon_N^{-1} =\O\big( \sqrt{N} \log N \big)$. 
  
Using the integral bound \eqref{integralbound}, we have a trivial bound for the integral on the right-hand side of 
\eqref{A4bd}, which concludes the proof. 
\end{proof}

\subsection{Proof of Theorem~\ref{thm:transmutation}}\label{sec:final proof}

Theorem~\ref{thm:transmutation} follows from a more refined result taking into account corrections in the single merging set. We state this first for completeness and then deduce Theorem~\ref{thm:transmutation} in a second step.

We assume that $\Phi \in L^\infty(\R^{2n})$ in~\eqref{def:psi} has a (fixed) compact support in $\D^n$ and consider the following H\"older norms, for any $\beta>0$
\[
 \|\Phi\|_{C^\beta} =  \inf\big\{ \eta>0  : |\Phi(\by)-\Phi(\bx)| \le \eta |\by-\bx|^\beta \text{ for all } \bx,\by\in\R^{2n} \big\} . 
\]
By extension, we also denote the uniform norm, 
\[
 \|\Phi\|_{C^0} = \inf\big\{ \eta>0  : |\Phi(\by)| \le \eta \text{ for a.e.}~\by\in\R^{2n}  \big\} .
\]
We use the notation \eqref{corrpot} to define refinements of the emerging scalar and vector potentials. For $j\in\{1,\dots, n\}$ and $\by \in \R^{2n}$, let
\begin{equation} \label{newV}
\mathbf{V}_j(\by) = 2N- \sum_{\ell\neq j}  N \mathrm{v}\left(\sqrt{N}(y_\ell-y_j)\right)
\end{equation}
and
\begin{equation} \label{newA}
\mathbf{A}_j(\by) 
= N y_j^\perp - \sum_{\ell \neq j} \left( \frac{(y_j-y_\ell) ^{\perp}}{|y_j-y_\ell|^2}  - \sqrt{N} \mathrm{a}\left(\sqrt{N}(y_1-y_\ell) \right) \right) . 
\end{equation}
Since $\mathrm{v}$ is a probability density function on $\R^2$, the second term in~\eqref{newV} is an approximation of the distributional potential $\sum_{\ell=1}^n  \delta_{(y_j-y_\ell)} $ for large $N$. On the other hand, recall that 
\begin{equation}\label{eq:corr vec estim}
\bigg|  \frac{y^\perp}{|y|^2} - \mathrm{a}(y)  \bigg| =  \bigg| \frac{y^\perp}{|y|^2} \frac{e^{|y|^2}-1 -|y|^2}{e^{|y|^2}-1} \bigg| \le \frac12,
\end{equation}
which is an efficient bound for small $y$, while for large $y$, $a(y)$ decays super-exponentially.

We denote
\[
\mathcal{L}_j=  -\i \nabla_{y_j}   - q N y_j^\perp +   \mathbf{A}_j(\by) 
\]
to state 

\begin{theorem}[\textbf{Statistics transmutation with corrections}]\label{thm:with corr}\mbox{}\\
Fix $n \ge 2$ and
let $\Psi_{\Phi}$ be as in~\eqref{def:psi}-\eqref{def:Psiqh}, \eqref{eq:choice Phi} with $\B=N \in \N$ and $\Phi \in H^1 \cap L^\infty(\R^{2n})$ with a fixed compact support in~$\D_{1-\epsilon}^{n}$. Assume moreover that 
\begin{equation}\label{eq:stupid}
\sum_{j=1}^n \int_{\R^{2n}}   \left| \mathcal{L}_j(\by)\Phi(\by) \right|^2 \leq C N^{s} 
\end{equation}
for some (possibly very large) power $s>0$.

Then, choosing $\kappa$ and $\gamma$ large enough, for any $\beta\ge 0$, we have for all $j\in\{1,\dots, n\}$,
\begin{multline}\label{eq:with corr 1}
\int_{\R^{2(N+n)}} \left| \left(-\i \nabla_{y_j} - q N y_j^\perp \right) \Psi_\Phi(\by;\bx)\right|^2 d\by d\bx  \\= \int_{\R^{2n}}   \left| \mathcal{L}_j(\by)\Phi(\by) \right|^2 \d\by +  \int_{\R^{2n}} |\Phi(\by)|^2 \mathbf{V}_j(\by) d\by  + \mathrm{Err} (\Phit) 
\end{multline}
with an error satisfying 
\begin{multline}\label{eq:with corr 2}
 \left|  \mathrm{Err} (\Phit)\right| \leq C_{n,\epsilon} \|\Phi\|_{C^\beta}^2 N^{-1/2 - \beta} (\log N)^{2 + \beta} \\ + C_{n,\epsilon}\|\Phi\|_{C^\beta} N^{-1/2 - \beta /2} (\log N)^{1+\beta/2} \sqrt{\int_{\DROPn^R}   \left| \mathcal{L}_j(\by)\Phi(\by) \right|^2}
\end{multline}
where we wrote 
\begin{equation}\label{eq:split config}
\DROPn = \DROPn^{\varnothing} \bigcup \left( \bigcup_{1\leq i\neq j \leq n } \DROPn^{ij,\gamma} \right) \bigcup \DROPn^R    
\end{equation}
as a disjoint union.
\end{theorem}

\begin{proof}
 Without loss of generality we fix $j=1$. We start from the exact calculation from Proposition~\ref{prop:exp}, which we recall:
\begin{align} \nonumber
\int_{\R^{2(N+n)}} \left| \left(-\i \nabla_{y_1} - q N y_1^\perp \right) \Psi_\Phi(\by; \bx)\right|^2 d\bx d\by 
 & =  \int_{\R^{2n}}   \left| \left(-\i \nabla_{y_1} + \mathcal{A}_1 - q N y_1^\perp \right) \Phi(\by) \right|^2   d\by \\
&\quad \label{eq:exact}
 + \int_{\R^{2n}} |\Phi(\by)|^2 \mathcal{V}_1(\by) d\by. 
 \end{align}
In the right-hand side the integration is by assumption restricted to $\DROPn$ as defined in~\eqref{eq:config}. 
The set $\DROPn^R$ with multiple mergings or extremely close single mergings (the exponent $R$ stands for ``remainder of the configuration space'') clearly satisfies
\begin{equation} \label{Rdomain}
\DROPn^R  \subset  \big\{ \exists \{i,j\} : |y_i - y_j|^2  \le N^{-1-\gamma} \big\}
\cup
 \big\{ \exists \{i,j\} \neq \{k,\ell\},  |y_i-y_j| ,|y_k-y_\ell|  \le  2 \delta_N\big\} . 
\end{equation}
where $\{i,j\}$ are (non-ordered) pairs of indices $i,j\in\{1,\dots, n\}$ with $i\neq j$. 
Hence, if $\gamma\ge 1$, there exists a constant $C_n$ so that  the volume of $\DROPn^R$,
\begin{equation} \label{eq:Rvol}
|\DROPn^R| \le C_n \delta_N^4 . 
\end{equation}
In the following we are at liberty to choose two fixed constants $\kappa >0$ and $\gamma >0$. For the contributions from the scalar potential, second term in~\eqref{eq:exact}, we bound 
\begin{equation}\label{eq:split V}
 \left|\int_{\R^{2n}} |\Phi|^2 \left(\mathcal{V}_1 - \mathbf{V}_1 \right) \right| \leq  \sum_{\cD} \int_{\cD} \left|\mathcal{V}_1 - \mathbf{V}_1 \right| |\Phi|^2 
\end{equation}
where the sum runs over all subsets $\cD$ from the decomposition~\eqref{eq:split config}. To deal with the vector potential we have, expanding squares and using Schwarz's inequality
\begin{align*}
 \Big|\int_{\R^{2n}}   \left| \left(-\i \nabla_{y_1} + \mathcal{A}_1 - q \B y_1^\perp \right) \Phi \right|^2 &- \int_{\R^{2n}}   \left| \mathcal{L}_1  \Phi \right|^2 \Big| \leq C \sum_\cD \varepsilon_\cD \int_{\cD} \left| \mathcal{L}_1 \Phi \right|^2 \nonumber\\
 &+ C \sum_\cD \left( 1 + \varepsilon_\cD^{-1}\right) \int_{\cD} \left|\mathcal{A}_1 - \bA_1 \right|^2 |\Phi|^2. 
\end{align*}
The sums are again over subsets $\cD$ from the decomposition~\eqref{eq:split config}. Parameters $\varepsilon_\cD$ have been inserted that we optimize over to obtain 
\begin{align}\label{eq:split A}
\Big|\int_{\R^{2n}}   \left| \left(-\i \nabla_{y_1} + \mathcal{A}_1 - q \B y_1^\perp \right) \Phi \right|^2 &
 \leq  C \sum_\cD E(\cD)^{1/2} \left(E(\cD)^{1/2} + \left(\int_{\cD} \left| \cL_1 \Phi \right|^2\right)^{1/2} \right) 
\end{align}
where we denote 
$$ E(\cD) := \int_{\cD} \left|\mathcal{A}_1 - \bA_1\right|^2 |\Phi|^2.$$
We now discuss separately the contributions to the sums over $\cD$ above. 

\medskip

\noindent\textbf{Contribution from the no-merging set.} It follows from Proposition~\ref{prop:nomerg}, elementary estimates on the correction functions~\eqref{corrpot} and the $L^2$-normalization of $\Phi$ that for any $\kappa>0$,
\begin{align*}
 \int_{\DROPn^\varnothing} \left|\mathcal{V}_1 - \mathbf{V}_1 \right| |\Phi|^2 &\leq C N^{2-2\kappa^2}\\
 E(\DROPn^\varnothing) \leq C N^{2-4\kappa^2}.
\end{align*}

\medskip 

\noindent\textbf{Contribution from single mergings.} We deal with the contribution from a set $\DROPn^{ij,\gamma}$ using Proposition~\ref{prop:merg}. This gives, for any $i\neq j$
$$
\int_{\DROPn^{ij,\gamma}} \left|\mathcal{V}_1 - \mathbf{V}_1 \right| |\Phi|^2 \leq C N^{1+3\gamma - 2 \kappa^2}\\
$$
and
$$
 E \left(\DROPn^{ij,\gamma}\right) \leq C N^{2+4\gamma - 4 \kappa^2} 
$$
using again the $L^2$-normalization of $\Phi$.

Moreover, by~\eqref{eq:stupid},  $\int_{\cD} \left| \cL_1 \Phi \right|^2$ is bounded by $O(N^s)$ in~\eqref{eq:split A}.  Hence the error coming from the sets $\DROPn^\varnothing$ and $\DROPn^{ij,\gamma}$ are all made negligible by choosing $\kappa$ sufficiently large (depending on $\gamma,s$).

\medskip 

\noindent\textbf{Contributions from the remainder.} It follows from~\eqref{Rdomain} and~\eqref{eq:Rvol} that, since $\Phi$ is assumed to be $\beta$-H\"older continuous,  
\begin{equation} \label{Rbound}
\int_{\DROPn^R} |\Phi(\by)|^2  \d\by \le C_n \|\Phi\|_{C^\beta}^2 \delta_N^{4+2\beta} \leq C_n \|\Phi\|_{C^\beta}^2 \left(\frac{\kappa ^2 \log N}{N}\right)^{2+\beta}. 
\end{equation}
On the other hand, using that $v$ is uniformly bounded and~\eqref{eq:corr vec estim} it follows from the definitions~\eqref{corrpot} that
\begin{equation} \label{VAcontrol}
0\le \mathbf{V}_j(\by) \le 2N 
\qquad\text{and}\qquad
|\mathbf{A}_j(\by)  - N y_j^\perp| \le n \sqrt{N}/2.
\end{equation}
Combining the above with~\eqref{eq:Vglob} and the triangle inequality we find 
\[
\int_{\DROPn^R} \left|\mathcal{V}_1 - \mathbf{V}_1 \right| |\Phi|^2 \leq C \|\Phi\|_{C^\beta}^2 N^{-1/2 - \beta}  (\log N) ^{2+\beta}
\]
where the constant depends on $\kappa$ and $n$. Similarly, using~\eqref{eq:Aglob2} and~\eqref{VAcontrol},
\[
E(\DROPn^R) = \int_{\DROPn^R} \left|\mathcal{A}_1 - \mathbf{A}_1 \right|^2 |\Phi|^2 \leq C \|\Phi\|_{C^\beta}^2 N^{-1 - \beta}  (\log N) ^{2+\beta}
\]

\medskip 

\noindent\textbf{Conclusion.} Putting everything together and inserting in~\eqref{eq:split V}-\eqref{eq:split A} yields our final bound. For simplicity we use~\eqref{eq:stupid} and take $\kappa$ so large that all the remainders coming from the no-merging and single-merging regions can be included in the main error. The latter only comes from our treatment of the remaining domain $\DROPn^R$. 
\end{proof}

We now conclude the

\begin{proof}[Proof of Theorem~\ref{thm:transmutation}]
In all this proof we take $\kappa$ large enough but fixed, and we stick for simplicity to the case $\beta = 1$ from Theorem~\ref{thm:with corr}, our starting point. There remains to extract and estimate the contributions of the correction fields~\eqref{corrpot} to the right-hand side of~\eqref{eq:with corr 1}. We have that 
\begin{multline*}
\int_{\R^{2n}}   \left| \mathcal{L}_j(\by)\Phi(\by) \right|^2 \d\by +  \int_{\R^{2n}} |\Phi(\by)|^2 \mathbf{V}_j(\by) d\by =  2N + \int_{\R^{2n}}   \left| \left(-\i \nabla_{y_j}+ \bAt (y_j) \right) \Phi \right|^2 \\ 
+  \int_{\R^{2n}} \left( \sum_{k\neq j}  N \mathrm{v} \left(\sqrt{N} (y_j-y_k)\right) \right) |\Phi(\by)|^2 d\by + \int_{\R^{2n}} \left| \sum_{k\neq j} \sqrt{N} \mathrm{a} \left(\sqrt{N} (y_j-y_k)\right) \right| ^2 |\Phi|^2  \\
+ 2 \Re \int_{\R^{2n}} \overline{\Phi}(\by) \sum_{k\neq j} \sqrt{N} \mathrm{a} \left(\sqrt{N} (y_j-y_k)\right) \cdot \mathcal{L}_j(\by) \Phi(\by)  d\by 
\end{multline*}
By Schwarz's inequality and using the condition~\eqref{eq:stupid pre}, the last term is estimated by 
\[
N^{-\kappa^2}  \sqrt{\int_{\DROPn^\varnothing}   \left| \mathcal{L}_j\Phi \right|^2 }
+  \sqrt{\int_{\DROPn\setminus\DROPn^\varnothing} \left| \sum_{k\neq j} \sqrt{N} \mathrm{a} \left(\sqrt{N} (y_j-y_k)\right) \right| ^2 |\Phi|^2  \int_{\DROPn\setminus\DROPn^\varnothing}   \left| \mathcal{L}_j\Phi \right|^2 } . 
\]
The first term coming from the no-merging set is made negligible by choosing $\kappa$ sufficiently large. Thus we have to bound 
$$ \int_{\R^{2n}} N \mathrm{v} \left(\sqrt{N} (y_j-y_k)\right) |\Phi(\by)|^2 d\by $$
and 
$$ \int_{\R^{2n}} \left| \sum_{k\neq j} \sqrt{N} \mathrm{a} \left(\sqrt{N} (y_j-y_k)\right) \right| ^2  |\Phi(\by)|^2 d\by .$$ 
The main contributions come from the single merging set, we leave it to the reader to bound other contributions using the explicit formulae~\eqref{corrpot}. 

Using that $\mathrm{v}$ is a probability density function and~\eqref{eq:Phi on merging}, we find 
\[\begin{aligned} \int_{|y_j-y_k| \leq \delta_N(\kappa)} N \mathrm{v} \left(\sqrt{N} (y_j-y_k)\right) |\Phi(\by)|^2 d\by  & \leq C_{\Phit}^2  \delta_N (\kappa)^2 \int_{\R^2} N \mathrm{v}(\sqrt{N}y) \frac{\d y}{\pi} \\
& = C_{\Phit}^2 \kappa^2 (\log N) N^{-1}
\end{aligned}\]
and using~ the bound \eqref{eq:corr vec estim}
 \begin{multline*}
 \int_{|y_j-y_k| \leq \delta_N(\kappa)} \left| \sqrt{N} \mathrm{a} \left(\sqrt{N} (y_j-y_k)\right) \right| ^2 |\Phi|^2 \\ \leq C \int_{|y_j-y_k| \leq \delta_N(\kappa)} \frac{1}{|y_j-y_k|^2} |\Phi|^2 \\
 \leq C_{\Phit} ^2 \delta_N (\kappa) ^2 = C_{\Phit} ^2 \kappa^2 (\log N) N^{-1}.
\end{multline*}
These estimates and~\eqref{eq:with corr 2} lead to~\eqref{eq:main estimate}. Note that the main errors, into which we absorbed all the others, are the contributions of the correcting vector field $\mathrm{a}$ in the single-merging set. 
Then~\eqref{eq:main estimate gauge} follows by changing gauge as discussed in the remarks following Theorem~\ref{thm:transmutation}. 
\end{proof}

\section{Applications}\label{sec:applications}

We now give more concrete applications of Theorem~\ref{thm:transmutation}, i.e. we evaluate the error terms in~\eqref{eq:main estimate gauge} for natural choices of $\Phit$ in situations of physical relevance. Before doing that, we precisely define in Section~\ref{sec:delta int} what is meant by the delta/contact interaction in our general Hamiltonian~\eqref{eq:intro hamil}.

As an illustration of applications of Theorem~\ref{thm:transmutation} we consider the case where the tracer particles are bosons, and are hence turned into fermions by the strong bath-tracer coupling. In this case, our main assumption~\eqref{eq:Phi on merging} is naturally imposed by the Pauli principle acquired after statistics transmutation. It is convenient to distinguish two cases. We first discuss the case of tracer particles having the same charge as bath particles ($q=1$). In this case the effective problem for tracers no longer depends on $\B$. We next discuss the case of unequal charges ($q\neq 1$), which requires a bit more care, for the effective problem does depend on $\B$, and it is in fact natural to restrict one-particle states available to the tracers to those of an effective lowest Landau level.     

\subsection{Delta interactions}\label{sec:delta int}
 
Here we define precisely the delta interaction we use in~\eqref{eq:intro hamil}, which is a slight generalization of operators discussed at length e.g. in~\cite{RouSerYng-13a,RouSerYng-13b,LewSei-09,SeiYng-20,RouYng-19}. 
 
The inter-species interaction in~\eqref{eq:intro hamil} acts on pairs of bath/tracer particles, it is thus sufficient to define $\delta(y-x)$ as a self-adjoint operator on 
$$ \gH^{1 \oplus 1} = L^2 (\R^2) \otimes \LLL$$
and check that the associated quadratic form is indeed as announced in~\eqref{eq:intro delta}.  
 
\begin{lemma}[\textbf{Bath-tracer delta interaction}]\label{lem:delta}\mbox{}\\
We identify $\R^{4} \ni (y;x) \leftrightarrow (w;z) \in \C^2$. For a function $\psi\in\LLL$ as in~\eqref{eq:intro LLL} we write 
$$\psi (x) = f(z) e^{-\frac{\B}{2} |z|^2}$$ 
with $f$ analytic .

We define the linear operator $\delta(y-x)$ by its action on tensor products $u\otimes \psi\in\gH^{1 \oplus 1}$
\begin{equation}\label{eq:def delta}
\left(\delta(y-x) u\otimes \psi \right) (y;x) = \frac{\B}{\pi} e^{-\B |w| ^2} e ^{\B z\overline{w}} e ^{-\frac{\B}{2}|z|^2} f(w) u (w).
\end{equation}
The action is then linearly extended to the full vector space $\gH^{1 \oplus 1}$, and defines a bounded self-adjoint operator with associated quadratic form
\begin{equation}\label{eq:delta quad}
\left\langle \Psi_{1 \oplus 1} | \delta (x - y) | \Psi_{1 \oplus 1} \right\rangle_{L^2} = \int_{\R^2} |\Psi_{1 \oplus 1} (x;x)| ^2 dx.   
\end{equation}
\end{lemma}
 
\begin{proof}
Using  the reproducing kernel~\eqref{def:K} of the lowest Landau level, we can rewrite \eqref{eq:def delta} as 
\[
\left(\delta(y-x) u\otimes \psi \right) (y;x) =  u(w) \psi(w)  \K_\infty(z,w) . 
\]
It immediately follows from the symmetry of the kernel  $\K_\infty$ that the operator $\delta(y-x)$ acting on $\gH^{1 \oplus 1}$ is symmetric. 
Note that the right-hand side also lies in $\gH^{1 \oplus 1}$  because for $\psi \in \LLL$, 
$$ \sup_{w\in\C} |\psi (w)|  \leq  \frac{\B}{\pi} \norm{\psi}_{L^2} ,$$
as follows by using the reproducing Property~\ref{pro:reprod}, the Cauchy-Schwarz inequality and the uniform bound \eqref{reproducing3} (see~\cite{Carlen-91} for generalizations). 

Then, with $u,v\in L^2$ and $\psi,\varphi \in \LLL$, it holds
\[\begin{aligned}
\left\langle u\otimes \psi | \delta (x - y) | v \otimes \varphi \right\rangle_{L^2} = \frac{\B}{\pi} 
\int_{\C} \overline{u(w)}  v(w) \varphi(w) \bigg( \int_{\C}  \overline{\psi(z)}  \K_\infty(z,w) dz \bigg) d w 
\end{aligned}\]
Using again the reproducing Property~\ref{pro:reprod} to perform the inner integral, we obtain 
$$
\left\langle u\otimes \psi | \delta (x - y) | v \otimes \varphi \right\rangle_{L^2} = \int_{\C} \overline{u (w) \psi(w)}  v(w) \varphi(w) d w = \int_{\R^2} u\otimes \psi (y;y) v\otimes \varphi (y;y) d y.   
$$
Thus~\eqref{eq:delta quad} is proved and our operator is indeed bounded. 
Let us also observe that
\begin{equation*}\label{eq:projector}
 \delta(y-x) ^2 = \frac{\B}{\pi} \delta(y-x).  \qedhere
\end{equation*}
\end{proof}
 
From the above definition it is clear that the interaction term in~\eqref{eq:intro hamil}
\begin{equation}\label{eq:interaction}
g \sum_{k=1} ^N \sum_{j=1} ^n \delta (x_k - y_j)
\end{equation}
is a positive operator. Its' lowest eigenvalue is $0$ (achieved in particular by states of the form~\eqref{def:psi}-\eqref{def:Psiqh}), and is separated by a gap from the rest of the spectrum (the dependence of this gap on $n$ and $N$ is an open problem~\cite{NacWarYou-20a,NacWarYou-20b,Rougerie-xedp19,RouSerYng-13b,WarYou-21}). For large $g>0$ it is relevant to restrict available states to the ground eigenspace of~\eqref{eq:interaction}. While~\eqref{def:psi}-\eqref{def:Psiqh} does not exhaust this space it gives a simple tractable form to work with. 
Specifically, we have assumed the quasi-holes to be simple, corresponding to $p=1$ (and $\mu=1$) in~\eqref{def:Psiqh frac}, which is motivated by an otherwise cost (due to radial confinement) of increased angular momentum of bath and size of droplet.
It also minimizes the constant term in~\eqref{eq:main estimate frac}.
We refer to~\cite{LieRouYng-17,RouYng-19,OlgRou-19} for results about the emergence of quasi-holes as in the ansatz~\eqref{def:psi}, and to~\cite{Rougerie-Elliott,Rougerie-Ency} for reviews thereof.

\subsection{Case of equal charges}

In this subsection we set $q=1$. In view of Theorem~\ref{thm:transmutation}, the natural effective problem for the impurities is governed by the Hamiltonian 
\begin{equation}\label{eq:eff q=1}
\Heff_n = \sum_{j=1}^n - \frac{1}{2m} \Delta_{y_j} + W(y_1,\ldots,y_n) 
\end{equation}
where $W$ is the potential appearing in~\eqref{eq:intro hamil}. For simplicity of exposition we assume that $W$ is a smooth bounded function, but there is room in our method to accomodate some reasonable singularities. 

To avoid boundary effects we confine tracer particles strictly within the droplet of bath particles. Let hence $\alpha < 1$ and 
\begin{equation}\label{eq:disk alpha}
\D_\alpha := \left\{ x \in \R^2, |x| \leq \alpha \right\}. 
\end{equation}
We impose Dirichlet boundary conditions on $\partial \D_\alpha$ to define
\begin{equation}\label{eq:Eeff}
\Eeff (n) := \inf\left\{ \langle U_n | \Heff_n \, U_n \rangle : U_n \in H^1_0 (\D_\alpha ^n), \int_{\D_\alpha ^n} |U_n| ^2 = 1, U_n \ \mbox{antisymmetric}   \right\}
\end{equation}
the appropriate \emph{fermionic} ground-state energy of $\Heff$. We have 

\begin{corollary}[\textbf{Statistics transmutation, equal charges}]\label{cor:equal}\mbox{}\\
Set $q=1$ and let $E(n \oplus N)$ be the lowest eigenvalue of $H_{n \oplus N}^W$ (cf.~\eqref{eq:intro hamil}) acting on $\gH^{n \oplus N}_{\rm sym}$ (defined in~\eqref{eq:intro full space}). Fix $\alpha <1$ and set $\B=N$. We have 
\begin{equation}\label{eq:result q=1}
E(n \oplus N) \leq \frac{n\B}{m} + \Eeff (n) + C_n \frac{\log N}{N} .
\end{equation}
\end{corollary}

\begin{proof}
We apply Theorem~\ref{thm:transmutation} with $\Phit$ a minimizer for~\eqref{eq:Eeff}. 
Then,  $\Phit$ solves 
\begin{equation}\label{eq:vareq}
\Heff_n \Phit = \Eeff (n) \Phit
\end{equation}
and it follows from standard elliptic regularity theory ~\cite{Evans-98,GilTru-01} that $\Phit$ is a smooth function, with bounds independent on $\B$. In particular, both conditions \eqref{eq:Phi on merging} and  \eqref{eq:stupid pre} hold with a fixed constant $C_{\Phit}$ and $s=0$ since $q=1$.
Hence, the estimate \eqref{eq:result q=1} follows from Theorem~\ref{thm:transmutation} if we show that the last error term in 
\eqref{eq:main estimate gauge 2}, 
\[
\int_{ \DROPn\setminus \DROPn^{\varnothing}} \left| \nabla_{y_j}\Phit(\by) \right|^2 \d\by \le C_n  \frac{\log N}{N} . 
\]
This bound follows directly from the smoothness of  $\Phit$ and the fact that the measure 
$$|\DROPn\setminus \DROPn^{\varnothing}| \le C_n \delta_n^2 \leq C_n \kappa^2 \frac{\log N}{N} $$
for some fixed constant $C_n$; cf.~\eqref{eq:delta}--\eqref{eq:nomerg}. 
\end{proof}

\subsection{Case of unequal charges} 

We outline the analysis of a case where the charge of tracer particles differs from that of the bath particles. The interested reader should be able to turn the following considerations into a complete proof.

In view of Theorem~\ref{thm:transmutation}, the effective Hamiltonian for the impurities is given by
\[
\Heff_n = \sum_{j=1}^n \frac{1}{2m} \left( -\i \nabla_{y_j} - (q-1)\B y_j ^\perp )\right)^2 + W(y_1,\ldots,y_n). 
\]
If $q\neq 1$ is fixed, the strong constant magnetic field $\B \to\infty$  forces all tracer particles into the lowest Landau level associated to the first term, namely 
\begin{equation*}
\LLL_{|q-1|}:= \left\{ \psi \in L^2 (\R^2) : \psi(x) = f(z) e ^{-\frac{|q-1|\B}{2}|z|^2}, f \mbox{ analytic} \right\}. 
\end{equation*}
This is the space of functions such that 
\begin{equation}\label{eq:LLLq}
\left( -\i \nabla_{y} - (q-1)\B y ^\perp )\right)^2 \psi = 2 |q-1| \B \psi,
\end{equation}
separated from the rest of the spectrum by a spectral gap of order $\B$. It is natural to choose $\Phit$ in our trial state to be made entirely of orbitals from this space. Then 
\begin{equation}\label{eq:eff q} 
\langle \Phit | \Heff_n \Phit \rangle = \frac{|q-1| \B n }{m} + \langle \Phit | W(y_1,\ldots,y_n) \Phit \rangle
\end{equation}
and we should seek to minimize the second term, under the constraint that 
$$ \Phit \in   \LLL_{|q-1|}^{\otimes_{\mathrm{asym}} n }.  $$
In the limit $\B \to \infty$ with fixed $n$ (or, in fact, essentially as long as $n\ll \B$), this reduces to a classical mechanics problem~\cite{LieSolYng-94,LieSolYng-94b,LieSolYng-95} where particles are forced into describing small cyclotron orbits around guiding centers whose locations are determined by the external potentials (see e.g.~\cite{ChaFlo-07,ChaFlo-09,ChaFloCan-08,Goerbig-09,GoeLed-06,Jain-07} for general discussions and references in the physics literature). 
For simplicity we could consider a bona-fide two-body problem with regular potentials $V$ and $w$, that is,
\begin{equation}\label{eq:2body}
W(y_1,\dots,y_n)= \sum_{j=1} ^n V (y_j) + \sum_{1\leq j < k \leq n} w(y_j - y_k) .
\end{equation}
We will sketch the proof of the following 

\begin{corollary}[\textbf{Case of unequal charges}]\label{cor:unequal}\mbox{}\\
Set $q\neq 1$ and let $E(n \oplus N)$ be the lowest eigenvalue of $H_{n \oplus N}^W$ (defined in~\eqref{eq:intro hamil}) acting on $\LLL_{|q-1|}^{\otimes_{\mathrm{asym}} n } \otimes \LLL ^{\otimes_{\mathrm{asym}} N }$ with $W\in L^1(\R^{2n})$, symmetric and continuous.

Let $\bR_n = (R_1,\ldots,R_n) \in \D^{n}$ be fixed distinct points. We have 
\begin{equation}\label{eq:result q neq 1}
E(n \oplus N) \leq \frac{\left( |q-1| + 1 \right) \B n }{m} + W \left(R_1,\ldots,R_n\right) + o(1)
\end{equation}
in the limit $\B = N \to \infty$. 
\end{corollary}

Note that 
\begin{multline}\label{eq:intro hamil bis}
H_{n \oplus N}^W = g \sum_{k=1} ^N \sum_{j=1} ^n \delta (x_k - y_j) + \sum_{j=1}^n \left\{ \frac{1}{2m}\left( -\i \nabla_{y_j} - q \B y_j ^{\perp} \right)^2 \right\} + W(y_1,\ldots,y_n) \\ \geq \frac{n|q| \B}{m} + \inf_{\R^{2n}} W
\end{multline}
as an operator, the lower bound being simply the lowest eigenvalue of the second term plus the infimum of the potential. Choosing the points $R_1,\ldots,R_n$ appropriately in Corollary~\ref{cor:unequal} we have
$$ \left\langle \Psi_{\Phi}, \, H_{n \oplus N}^W \Psi_{\Phi} \right\rangle \leq \frac{\left( |q-1| + 1 \right) \B n }{m} + \inf_{\R^{2n}} W + o(1).$$
It is noteworthy that this upper bound matches the trivial lower bound~\eqref{eq:intro hamil bis} up to $o(1)$ in the special case $q\geq 1$. In this case one can, \emph{independently of $g>0$} and with a very good precision, jointly minimize the three (non-commuting) terms of the Hamiltonian. The expectation of the first term in~\eqref{eq:intro hamil bis} is as small as possible, namely $0$. We recover, modulo small errors, the lower bound $\frac{n|q| \B}{m}$ coming from the second term. The next term in the expansion is the classical energy associated to $W$.

\begin{proof}[Sketch of proof for Corollary~\ref{cor:unequal}]
We recall~\cite{ChaFlo-07,RouYng-19} the definition of ``vortex coherent states'', suitable to describe the quantum motion of a particle on a cyclotron orbit compatible with~\eqref{eq:LLLq} centered at $R\in \R^2$: in complex notation $y,R \leftrightarrow w,Z \in \C$
\begin{equation}\label{eq:vortex state}
\psi_R (w) := \sqrt \frac{|q-1|\B}{\pi} e^{- \frac{|q-1|\B}{2}\left( |w| ^2 + |Z|^2 - 2 \overline{Z} w \right)}.
\end{equation}
Note that $\psi_R (w)  =  \K_\infty(w,Z)  $ according to \eqref{def:K} after suitably adjusting $\B$ and 
\begin{equation}\label{eq:vortex state mod}
\left|\psi_R (w)\right|^2 =   \frac{ |q-1|\B}{\pi} e^{- |q-1| \B |w-Z| ^2}
\end{equation}
so that, as $\B\to\infty$, the above is very much concentrated around $Z$. 

The function entering our trial state's definition is
\begin{equation}\label{eq:Slater q}
\Phit(\by) := \frac{c (\bR_n)}{\sqrt{n!}} \det \left( \psi_{R_j} (w_k) \right)_{1\leq j, k \leq n} .
\end{equation}
 Using the reproducing property \eqref{reproducing1} of $\K_\infty$,  the appropriate  normalization constant is
\[\begin{aligned}
c(\bR_n)^{-2}  & = \frac{1}{n!} \int_{\C^n} \big| \det \big( \psi_{R_j} (w_k) \big)_{1\leq j, k \leq n} \big|^2 d\bw  \\
& =  \det \big( \psi_{R_j} (R_k) \big)_{1\leq j, k \leq n} . 
\end{aligned}\]
Hence, for a fixed $\bR_n \in \R^{2n}$, 
$c(\bR_n)\to1$  exponentially fast as $\B\to\infty$.
Strictly speaking the state $\Phit$ is not supported in $\DROPn$, but for a fixed $\bR_n = (R_1,\ldots,R_n) \in \D^{n}$ all distinct, up to a  $\O(N^{-\infty})$ error (in $L^2\cap L^\infty$) as $N\to\infty$, $\Phit$ is supported in the small set
\begin{equation} \label{qset}
\bigcup_{\sigma\in \S_n} \big\{ \by\in \R^{2n} :  |y_{\sigma(j)}-R_j| \le \delta_N ;\, j=1,\dots, n \big\} 
\end{equation}
where
$\delta_N(\kappa)= \kappa \sqrt{\frac{\log N}{N}}$ as in \eqref{eq:delta} and $\S_n$ denotes the permutation group of $\{1,\dots, n\}$.
We also emphasize that since the points $R_j$ are distinct, 
$\eqref{qset} \subset \DROPn^{\varnothing}$ (the no-merging set).

Hence, up to a truncation, we can apply Theorem~\ref{thm:transmutation}. 
Condition \eqref{eq:stupid pre} follows from the fact that applying \eqref{eq:LLLq} to the states \eqref{eq:vortex state}, we have
\[\begin{aligned}
\int_{\R^{2n}}  \left| \left(-\i \nabla_{y_j}   - (q-1)\B y_j^\perp \right) \Phit \right|^2 
 & =  \int_{\R^{2n}} \overline{ \Phit } \left( -\i \nabla_{y} - (q-1)\B y ^\perp )\right)^2   \Phit  \\
& = 2 |q-1| \B  \int_{\R^{2n}}  |\Phit|^2  = 2|q-1| N ,
\end{aligned}\]
where we used multi-linearity of the Slater determinant \eqref{eq:Slater q}.

Using $\Phit \in \LLL_{|q-1|}^{\otimes_{\mathrm{asym}} n }$ as our trial state in Theorem~\ref{thm:transmutation}, by \eqref{eq:eff q} we obtain 
\begin{equation}\label{eq:result q neq 1 pre}
E(n \oplus N) \leq \frac{\left( |q-1| + 1 \right) \B n }{m} + \int_{\R^{2n}} W(\by) |\Phit (\by)|^2 d\by + \mathrm{Err} (\Phit)
\end{equation} 
where the error is controlled by \eqref{eq:main estimate gauge 2}. 
We will now argue that this error is negligible as $N\to\infty$.

The condition \eqref{eq:Phi on merging} plainly holds, but the trial state $\Phit$ depends on the magnetic field $\B=N$, so it is not a priori clear how the constant $C_{\Phit}$ behaves as $N\to\infty$ with its present definition. 
Note that in the proof of Theorem~\ref{thm:transmutation}, we only rely on the condition \eqref{eq:Phi on merging} to control the approximation errors when at least two particles are merging (i.e.~on the complement of  $\DROPn^{\varnothing}$). 
Since $\Phit$ (and its derivatives by analyticity) is essentially supported on \eqref{qset}, we have $C_{\Phit} = \O(N^{-\infty})$  as $N\to\infty$ when we restrict to the merging set (since we can make the parameter $\kappa$ arbitrary large). 
Similarly, the second term in  \eqref{eq:main estimate gauge 2}, 
\[
\int_{ \DROPn\setminus \DROPn^{\varnothing}} \left|\left(-\i \nabla_{y_j}   - (q-1)\B y_j^\perp \right)  \Phit \right|^2 = \O(N^{-\infty}) . 
\]
This argument shows that for this particular trial state
$ \mathrm{Err} (\Phit) = \O(N^{-\infty})$ as $N\to\infty$. 

To conclude the proof, it remains to compute the limit of the second term on the right-hand side of \eqref{eq:result q neq 1 pre}. 
By \eqref{eq:Slater q}, we can write
\begin{equation} \label{Rkernelq}
|\Phit(\by)|^2 = \frac{c (\bR_n)^2}{n!} \det \left( \sum_{k=1}^n  \psi_{R_i} (y_k) \overline{\psi_{R_j} (y_k)} \right)_{1\leq i,j \leq n}
\end{equation}
Using again that the orbital $\psi_R$ is exponentially concentrated on a neighborhood of $R$ and that the points $R_j$ are distinct, the matrix on the right-hand side of ~\eqref{Rkernelq} is diagonal up to an exponentially small correction (in $L^2\cap L^\infty$) as $N\to\infty$. Thus
\[ \begin{aligned}
|\Phit(\by)|^2 & = \frac{c (\bR_n)^2}{n!}  \prod_{j=1}^n \left( \sum_{k=1}^n  |\psi_{R_j} (y_k) |^2 \right) + \O(N^{-\infty}) \\
& = \frac{c (\bR_n)^2}{n!}   \sum_{k_1, \cdots ,k_n = 1}^n \prod_{j=1}^n |\psi_{R_j} (y_{k_j}) |^2 + \O(N^{-\infty})
\end{aligned}\]
and if the potential $W\in L^1(\R^{2n})$,  then  as $N\to\infty$
\[
\int_{\R^{2n}} W(\by) |\Phit (\by)|^2 d\by 
=  \frac{c (\bR_n)^2}{n!}   \sum_{k_1, \cdots ,k_n = 1}^n \int_{\R^{2n}} W(\by)  \prod_{j=1}^n |\psi_{R_j} (y_{k_j}) |^2
 +  \O(N^{-\infty})  . 
\]

By  \eqref{eq:vortex state mod}, $|\psi_R (w)|^2 \to \delta(w-R)$ in $C^*$ as $\B= N \to\infty$, so that if $W \in C(\R^{2n})$ and $R_j$ are distinct points, we conclude that
\[
\int_{\R^{2n}} W(\by) |\Phit (\by)|^2 d\by \to \frac{1}{n!}   \sum_{\sigma\in\S_n} W(R_{\sigma(1)}, \dots, R_{\sigma(n)}) . 
\]
This completes the proof. 
\end{proof}

 \newpage
 
\appendix

\section{Proof of Theorem~\ref{thm:exp}} \label{Appendix}

We give a slight variant of the proof from~\cite{Lambert-20}. Introduce a quantum wave-function for $N$ fermions
$$ \Psi_N (\z) = c_N \prod_{1\leq k<j \leq N} (z_k-z_j) e ^{-\frac{\B}{2}\sum_{j=1}^N |z_j|^2}$$
with $c_N= \cqh(\emptyset)$ a $L^2$ normalization constant (indeed $\Psi_N(\z) = \Psi_{\rm qh}(\emptyset, \z)$ in the absence of quasi-holes). With the one-body orbitals~\eqref{eq:orbitals} we have 
$$ \Psi_N (\z) = \frac{1}{\sqrt{N!}} \underset{N\times N}{\det} \left( \varphi_k (z_j) \right)$$
the Slater determinant made of the $N$ first orbitals ($k\in\{0,\dots,N-1\}$). It follows that 
$$ c_N^2 = \frac{1}{\pi^N \prod_{k=1}^{N} k! }\B^{N(N+1)/2}.$$
For a general fermionic (antisymetric in its arguments) $N$-body wave-function $\Phi_N$ on $L^2(\R^{dN})$ we define (the integral kernel of) its $n$-body density matrix, for $n\in\N$,
\begin{multline}\label{eq:red mat}
\gamma_{\Phi_N} ^{(n)} (x_1,\ldots,x_n;y_1,\ldots,y_n) \\ 
= { N \choose n} \int_{X_{N-n}\in \R^{d(N-n)}} \overline{\Phi_N (y_1,\ldots,y_n,X_{N-n})} \Phi_N (x_1,\ldots,x_n,X_{N-n}) dX_{N-n} 
\end{multline}
and its $n$-body density 
$$
\rho_{\Phi_N} ^{(n)} (x_1,\ldots,x_n) = \gamma_{\Phi_N} ^{(n)} (x_1,\ldots,x_n;x_1,\ldots,x_n).
$$
In probabilistic terms, these are the correlations of the (fermionic) point process. 
With the above conventions, rewriting  \eqref{def:Psiqh} in terms of $\Psi_{N+n}(\w, \z)$, it holds 
$$
\cqh(\w)^{-2} =  \int_{\C^N} \left| \Psi_{\rm qh}(\w, \z)\right|^2\d\z =  {N+n \choose n}^{-1} c_{N+n}^{-2} \,\rho_{\Psi_{N+n}} ^{(n)} (w_1,\ldots,w_n)  \left| \triangle(\w)\right|^{-2} \prod_{j=1}^n e^{\B|w_j|^2} .
$$
The claimed result is then the identity 
$$ \rho_{\Psi_{N+n}} ^{(n)} (w_1,\ldots,w_n) = \frac{1}{n!}\det_{n\times n} \left[\K_{N+n}(w_i,w_j) \right],$$
which follows from Wick's theorem for fermionic quasi-free states. 
Note that 
$$
\frac{1}{n!}   {N+n \choose n}^{-1} c_{N+n}^{-2} =  N! \pi^{N+n} {\textstyle  \prod_{k=1}^{N+n-1} k! } \, \B^{-(N+n)(N+n+1)/2} . 
$$
For the convenience of the reader we recall and prove it in the case of Slater determinants, relevant for the above. 

\begin{lemma}[\textbf{Density matrices of a Slater determinant}]\mbox{}\\
Let $u_1,\ldots,u_N$ be orthonormal functions in $L^2 (\R^{d})$. Let 
$$ \Phi_N (x_1,\ldots,x_N) = \frac{1}{\sqrt{N!}} \underset{N\times N}{\det} \left( u_k (x_j) \right)$$
be the associated Slater determinant. Its reduced density matrices~\eqref{eq:red mat} satisfy
\begin{equation}\label{eq:red mat Sla}
\gamma_{\Phi_N} ^{(n)} (x_1,\ldots,x_n;y_1,\ldots,y_n) = \frac{1}{n!} \det_{n\times n} \left[\gamma^{(1)}_{\Phi_N} (x_i,y_j) \right]
\end{equation}
and 
\begin{equation}\label{eq:red mat Sla 1} 
\gamma^{(1)}_{\Phi_N} (x,y) = \sum_{j=1} ^N \overline{u_j (y)} u_j (x).
\end{equation}
\end{lemma}

\begin{proof}
We denote, with $\S_n$ the permutation group of $\{1,\dots, n\}$, 
$$ u_1 \wedge \ldots \wedge u_n := \frac{1}{\sqrt{n!}} \underset{n\times n}{\det} \left( u_k (x_j) \right) = \frac{1}{\sqrt{n!}} \sum_{\sigma \in \S_n} \mathrm{sgn} (\sigma) \prod_{j=1}^n u_{j,\sigma(j)}$$
the Slater determinant constructed out of $n$ orthonormal functions. For a $L^2$ function $\Phi_n$ of $L^2_{\rm asym}(\R^{dn})$ we identify the orthogonal projector $|\Phi_n\rangle \langle \Phi_n |$ with its integral kernel $\overline{\Phi_n (Y_n)} \Phi_n (X_n)$. Likewise, we identify reduced density matrices~\eqref{eq:red mat} with the corresponding trace-class operators.

Then, a straightforward integration gives 
$$ \gamma_{\Phi_N} ^{(N-1)} = \sum_{\ell = 1} ^N \left| u_1 \wedge\ldots \wedge u_{\ell-1} \wedge u_{\ell+1} \wedge \ldots \wedge u_N \right\rangle \left\langle u_1 \wedge \ldots \wedge u_{\ell-1} \wedge u_{\ell+1} \wedge \ldots \wedge u_N \right|$$
using the orthogonality of $u_1,\ldots,u_N$. By induction we infer that 
$$ \gamma_{\Phi_N} ^{(n)} = \frac{1}{n!} \sum_{1\leq i_1 \neq \ldots \neq i_n \leq N} \left| u_{i_1} \wedge \ldots \wedge u_{i_n} \right\rangle \left\langle u_{i_1} \wedge \ldots \wedge u_{i_n}  \right|$$
and in particular this proves~\eqref{eq:red mat Sla 1}.

On the other hand, we can write the right-hand side of~\eqref{eq:red mat Sla} as 
$$ \frac{1}{n!} \det_{n\times n} \left[\sum_{j=1} ^N \overline{u_j (y)} u_j (x) \right] = \frac{1}{n!} \det_{n\times n} (A B)$$
with the matrices (respectively $n\times N$ and $N\times n$)
$$ A = \left( u_k (x_j)\right)_{1\leq j \leq n, 1\leq k \leq N}, \quad B = \left( \overline{u_k (y_j)}\right)_{1\leq j \leq N, 1\leq k \leq n}.$$
Comparing with the above, we see that~\eqref{eq:red mat Sla} is a consequence of the Cauchy-Binet formula. 
\end{proof}

\bibliographystyle{acm}

\end{document}